\documentclass[letterpaper]{article} 
\usepackage{aaai2026}
\nocopyright  
\usepackage{times}  
\usepackage{helvet}  
\usepackage{courier}  
\usepackage[hyphens]{url}  
\usepackage{graphicx} 
\usepackage{natbib}  
\usepackage{caption} 
\usepackage{algorithm}
\usepackage{algorithmic}
\usepackage{booktabs} 
\usepackage{array} 
\usepackage{amssymb} 
\usepackage[T1]{fontenc}
\usepackage{upquote} 
\usepackage{listings} 
\definecolor{promptrule}{gray}{0.45}
\lstdefinestyle{promptstyle}{%
  basicstyle=\ttfamily\footnotesize,
  frame=single,
  framesep=5pt,
  framerule=0.4pt,
  xleftmargin=5.4pt,
  xrightmargin=5.4pt,
  rulecolor=\color{promptrule},
  breaklines=true,
  breakatwhitespace=false,
  postbreak=\mbox{\textcolor{promptrule}{$\hookrightarrow$}\space},
  columns=fullflexible,
  keepspaces=true,
  showstringspaces=false,
  upquote=true, 
  extendedchars=true,
  aboveskip=9pt,
  belowskip=9pt,
  escapeinside={(*@}{@*)},
}

\title{Cherry-pick Override: LLM Judges Under-use\\the Non-Directional Verdicts Their Contract Authorizes}
\author{
    Haoran Xu
}
\affiliations{
    University of Glasgow\\
    2614067X@student.gla.ac.uk
}

\begin{document}
\maketitle

\begin{abstract}
Evidence-grounded fact verification increasingly uses LLM judges to turn
evidence into terminal verdicts. Many task contracts deliberately include
\emph{non-directional} verdicts---\textsc{Conflicting} for materially mixed
evidence, \textsc{Not Enough Evidence} for absent evidence---so that a system
can decline to assert a direction. Returning \textsc{Supports} or
\textsc{Refutes} on a claim whose authorized verdict is non-directional is
\emph{directional overcommitment}: a label error whose output is a direction
the contract did not authorize. Its conflict instance---both strands present,
one of them returned---is \textbf{Cherry-pick Override} (CPO); unlike the
insufficiency instance it cannot be repaired by retrieving more, because
nothing was missing.

On AVeriTeC's gold-\textsc{Conflicting} subset ($N_{\mathrm{C}}{=}150$) a
four-option typed panel commits directionally on $18.7\%$ of claims, and four
frozen contemporary judges (two proprietary, two open-weight) on
$21.3$--$48.0\%$, with $24.0$--$38.0\%$ on VitaminC-Mixed, a constructed
stress test on which the ranking reverses. These are conservative lower
bounds: the shared prompt specifies \textsc{Conflicting} in detail, forbids
collapsing to a direction, and demonstrates the label three times.

The failure is not specific to conflict: overcommitment on
gold-\textsc{Insufficient} cases runs at a comparable rate for every judge and
exceeds CPO for three of four on VitaminC-Mixed. Asked in a fresh context, two
judges identify the conflict on about two thirds of their own CPO claims,
against under $0.12$ where they committed correctly: the recognition is
available and does not reach the terminal action. Typed vocabulary, panel
voting, confidence thresholds, a structural validator and self-decomposition
each leave a substantial residual, and the last buys part of its reduction by
relabelling one-sided evidence as conflicting. Every rate is dataset-defined;
we report no human-confirmed prevalence.

\end{abstract}

\section{Introduction}\label{sec:intro}
Evidence-grounded fact-verification systems increasingly use LLM judges
to turn evidence into terminal verdicts. Their task contracts often do
more than list the directional outcomes: AVeriTeC offers
\textsc{Conflicting} and \textsc{Not Enough Evidence} alongside
\textsc{Supported} and \textsc{Refuted}
\citep{schlichtkrull2023averitec}. These two labels are
\emph{non-directional}. A contract that includes them has already decided
that some claims should not receive a direction, and has given the system
the vocabulary to say so. This paper asks what judges do with that
vocabulary.\footnote{Code, prompts, and all raw model outputs:
\url{https://github.com/HrxuAlbert/cherry-pick-override}. What is withheld,
and why, is at the end of \S\ref{sec:method}.}

Consider the benchmark claim \emph{``face masks cause hypoxia.''} Its
annotated evidence record contains a general refutation for healthy
individuals and a qualification concerning prolonged N-95 use by patients
with preexisting lung disease. The typed panel reads both strands yet
returns \textsc{Refutes} with mean confidence $0.96$. Under the benchmark
contract this is not merely an uncertain prediction. It is a directional
terminal decision on a case labelled for conflict preservation, taken
while the label that fits was available and unused.

Adjacent work studies how models detect or resolve conflicting evidence,
how conflict-aware QA systems communicate multiple answers, and when a
system should abstain \citep{jiayang2024econ,jin2024tugofwar,
nachshoni2025natconfqa,han2026evidenceconflicts}. These settings do not
isolate the same terminal-action error. In particular, models may identify
a conflict in their reasoning without allowing that recognition to change
their final prediction \citep{wallat2026facts}. Standard label accuracy
also merges this case with other mistakes, while confidence-based selective
prediction asks whether an output is reliable rather than whether a
directional output is licensed by the evidence structure
\citep{geifman2017selective,kamath2020selectiveqa}.

We study this as a contract-level diagnostic setting. \emph{Directional
overcommitment} is the event that a verifier returns \textsc{Supports} or
\textsc{Refutes} on a claim whose task-authorized verdict is
non-directional. It has two instances, one per non-directional label, and
they differ in what would repair them. When the authorized verdict is
\textsc{Conflicting}, both a supporting and a refuting strand are present
and the system returns one of them; we call this \textbf{Cherry-pick
Override} (CPO), after the selection it performs. When the authorized
verdict is \textsc{Insufficient}, no strand is present and the system
supplies one anyway. The insufficiency instance is a failure under
missing evidence and is in principle addressable by retrieving more; CPO
is a failure under complete evidence and is not, because nothing was
missing. That asymmetry, not a distinct trigger, is why the conflict
instance carries its own name here.

The definition is contract-relative. A three-option judge without
\textsc{Conflicting} lies outside the contract, even though its
directional rate on conflicting gold labels shows why exposing the label
matters. In contract, a typed panel still commits directionally on
$28/150$ AVeriTeC gold-\textsc{Conflicting} claims ($18.7\%$), and a
separately frozen replication with four contemporary individual judges,
two of them open-weight, finds non-zero CPO for every model on both
AVeriTeC and VitaminC-Mixed. These rates are conservative: the shared
prompt enumerates trigger conditions for \textsc{Conflicting}, forbids
collapsing such a case to a direction, and demonstrates the label three
times, so the residual is what survives an instruction designed to
remove it.

Characterizing the problem requires separating four effects. First, the
label vocabulary must expose the non-directional verdict; otherwise
directional output is partly forced by construction. Second, systems with
different commit rates must be compared on a common denominator and,
where relevant, at matched coverage. Third, a judge can be highly
confident in the direction it selected even when the evidence record
contains both support and refutation, so confidence and evidence
structure must be tested for complementary screening information without
treating either as ground truth. Fourth, the two instances must be
measured side by side: a rate observed only on conflicting cases cannot
distinguish a response to conflict from a general reluctance to use the
non-directional half of the vocabulary.

Our analysis addresses these requirements with an explicit four-label
contract, same-denominator metrics, paired-bootstrap comparisons, and a
diagnostic ladder covering typed vocabulary, panel aggregation,
confidence selection, and structural filtering. The ladder is not a
leaderboard and the final two-channel controller is not presented as a
solution. It is a reference probe for asking whether two imperfect
signals reject different cases. This framing keeps the empirical question
primary: does directional overcommitment remain after the task makes a
non-directional verdict available, and which common interventions fail to
remove it?

The paper contributes five elements. First, it defines directional
overcommitment as a contract-level failure with two instances, names the
conflict instance CPO, and gives a scope boundary separating it from
forced-choice output, generic abstention, hallucination, and ordinary
direction error (Section~\ref{sec:method}). Second, it provides a
same-denominator empirical characterization on AVeriTeC, including the
roles of panel aggregation and confidence, with VitaminC-Mixed as a
cross-dataset control, and tests whether the diagnosis persists across a
frozen panel of contemporary proprietary and open-weight individual
judges (Section~\ref{sec:results}). Third, it reports the insufficiency
instance on the same judges and substrates, and finds the failure is
\emph{not} specific to conflict: the two rates are comparable, and on
VitaminC-Mixed three of four judges overcommit more on
gold-\textsc{Insufficient} than on gold-\textsc{Conflicting}
(Section~\ref{sec:results}). Fourth, it separates recognition from
commitment: a fresh-context probe asks each judge whether the evidence is
materially two-sided, and on the two judges whose probe meets
pre-registered validity conditions the conflict is recognised on about two
thirds of that judge's own CPO claims, so the failure is not a failure to
see the conflict (Section~\ref{sec:recognition}). Fifth, it uses targeted
diagnostics to show
that structural and confidence signals are non-substitutable in the
tested settings, while explicitly bounding the claim where
matched-coverage intervals include zero (Section~\ref{sec:ablation}).

\section{Related Work}
\label{sec:related}

\paragraph{Evidence-grounded fact verification.}
Evidence-grounded verification benchmarks typically evaluate whether a system
returns the correct verdict together with adequate evidence.  AVeriTeC
provides a particularly relevant substrate: its four-way taxonomy includes
\textsc{Conflicting Evidence/Cherry-picking} alongside \textsc{Supported},
\textsc{Refuted}, and \textsc{Not Enough Evidence}, so a verifier is not
forced to choose a direction whenever its evidence is mixed
\citep{schlichtkrull2023averitec}.  The subsequent shared task evaluates the
joint quality of verdicts and retrieved evidence
\citep{schlichtkrull2024sharedtask}.  These evaluations are essential for
measuring fact verification, but ordinary label accuracy does not isolate the
action-relevant error considered here: issuing a directional terminal verdict
on a case for which preserving conflict is the authorized task outcome.

The substrate's own reliability bounds what any study built on it can claim.
AVeriTeC reports substantial inter-annotator agreement on verdicts,
$\kappa=0.619$, but does not report agreement broken down by verdict class
\citep{schlichtkrull2023averitec}. The reliability of the
\textsc{Conflicting} class specifically --- the class from which this paper's
primary denominator is drawn --- is therefore not established by the original
annotation effort, and an aggregate $\kappa$ on a four-way task whose mass sits
in \textsc{Supported} and \textsc{Refuted} can be substantial while a minority
class is much less reliable. We attempted an independent re-audit to close
that gap and did not complete it (\S\ref{sec:limits:contract}); every rate in
this paper is therefore dataset-defined, and we make no claim that the
\textsc{Conflicting} labels it rests on have been independently validated.

\paragraph{LLM judges as measuring instruments.}
This paper uses LLM judges, individually and as a majority-vote panel, as the
system under test. That instrument has documented failure modes of its own:
position, verbosity, and self-enhancement biases are all reported for
LLM-as-a-judge evaluation \citep{zheng2023mtbench}. Two consequences matter
here. Our panel aggregates three judges by majority vote, so any shared bias is
reinforced rather than averaged away, which is one reading of the amplification
we report in \S\ref{sec:results}. And because we never compare judges against
each other on the same output, the pairwise biases that dominate that
literature do not directly threaten our estimand, which is computed per judge
against a fixed dataset label.

\paragraph{Disagreement as signal rather than noise.}
The verdict this paper treats as authorized --- preserve the conflict rather
than resolve it --- is close to a label that records unresolved disagreement,
and there is a literature arguing that such disagreement is not annotation
error to be cleaned away. Pavlick and Kwiatkowski show that human disagreement
about textual inference persists as more ratings are collected and as context
is varied, and argue for evaluating against the distribution of plausible human
judgments rather than a single aggregate \citep{pavlick2019disagreements}.
Plank generalizes the position and surveys how human label variation should
enter data, modelling, and evaluation \citep{plank2022labelvariation}. We
inherit the methodological consequence, and it shaped what we report from the
calibration round we did run: pre-adjudication disagreement is described as a
measurement rather than as a defect that discussion should remove
(\S\ref{sec:limits:contract}).

\paragraph{Why a non-directional verdict might be under-used.}
Our results describe a behaviour without explaining it, and one candidate
explanation comes from outside verification. Kalai et al.\ argue that standard
training and evaluation reward a confident guess over an acknowledgement of
uncertainty, because accuracy-style metrics score an admitted unknown the same
as a wrong answer \citep{kalai2025hallucinate}. A contract that offers
\textsc{Conflicting} and \textsc{Not Enough Evidence} is exactly a setting
where that incentive would bite: both are non-directional, and under an
accuracy metric neither is rewarded more than a lucky direction. We do not test
this account --- it concerns training incentives we cannot observe --- but it
predicts the pattern we report in \S\ref{sec:results}, where under-use extends
to both non-directional verdicts rather than being specific to conflict.

\paragraph{Resolving conflict versus preserving conflict.}
Several RAG studies treat disagreement as a problem to resolve.  Jin et al.\
study conflicts between parametric and retrieved knowledge and introduce a
decoding method intended to resolve them \citep{jin2024tugofwar}.  Ge et al.\
study source-credibility information for fact checking under conflicting
evidence \citep{ge2025confact}, while Choi et al.\
learn a context-assessment signal to guide a RAG system toward the more
reliable source \citep{choi2025care}.  These methods address an important
regime in which source quality, time, scope, or additional evidence licenses
a directional resolution.  ECON directly evaluates conflict detection and
resolution and finds that LLMs can detect many conflicts yet often favour one
piece of evidence without justification \citep{jiayang2024econ}.  We study
the complementary contract in which preserving conflict is an available
terminal action.  The point is not to declare every disagreement unresolved,
but to diagnose a directional commitment on cases assigned to conflict
preservation.

\paragraph{Recognizing conflict and explaining uncertainty.}
Conflict-aware QA instead asks a system to expose disagreement in its answer.
NatConfQA requires models to identify specific conflicting answer pairs and
evaluates conflict-aware multi-answer responses \citep{nachshoni2025natconfqa}.
CLUE identifies evidence relations to explain sources of uncertainty in
automated fact checking \citep{sun2026clue}.  These works make conflict
recognition and communication explicit, but their output object is an answer
set or an uncertainty explanation rather than a typed fact-verification
terminal action.  Cherry-pick Override (CPO) specializes the question to a
verifier that has \textsc{Conflicting} available as a final label: did the
system nevertheless make an unauthorized \textsc{Supports}/\textsc{Refutes}
commitment?

\paragraph{From recognition to final prediction.}
The closest conceptual precursor is work showing that conflict recognition
does not necessarily affect a model's final output.  In temporal knowledge
conflicts, larger models can verbalize that a fact may have changed while
failing to propagate that judgment to their prediction
\citep{wallat2026facts}.  CPO turns this recognition--prediction gap into a
contract-specific terminal-action diagnostic: the output space contains a
typed conflict-preserving verdict, and the measured event is choosing a
direction instead.

\paragraph{Selective prediction and conflict-aware abstention.}
Selective prediction formalizes the risk--coverage trade-off induced by a
reject option \citep{geifman2017selective}; later QA work shows that raw
confidence alone can be unreliable for deciding when to abstain
\citep{kamath2020selectiveqa}.  For LLM systems, self-evaluation can provide
selective-generation scores \citep{ren2023selective}, and selective
LLM-judge pipelines can decide whether to trust a judge's decision or
escalate the evaluation to a stronger judge \citep{jung2025trust}.  Recent
work also evaluates an abstention score that
combines confidence with a detector of evidence conflict in controlled
biomedical RAG \citep{han2026evidenceconflicts}.  Consequently, we do not
claim that combining structural and confidence signals is novel.  Instead, we
test whether structural evidence provides incremental screening information
for benchmark-defined CPO events.  A routed \textsc{No-Commit} state is a
possible deployment response to that event, not a fifth verification label
and not the paper's primary contribution.

Table~\ref{tab:positioning} summarizes the distinction by evaluation target,
rather than by whether a system could in principle be adapted to another
setting.

\begin{table*}[t]
\centering
\small
\setlength{\tabcolsep}{5pt}
\renewcommand{\arraystretch}{1.0}
\caption{Positioning CPO among nearby settings.  ``--'' denotes an aspect
outside the cited setting's primary evaluation target, not an absent model
capability.}
\label{tab:positioning}
\resizebox{\textwidth}{!}{%
\begin{tabular}{>{\raggedright\arraybackslash}p{0.18\textwidth}
                >{\raggedright\arraybackslash}p{0.25\textwidth} c c c c}
\toprule
Systems & Primary objective & \begin{tabular}[c]{@{}c@{}}Resolve\\ toward one\\ answer\end{tabular} & \begin{tabular}[c]{@{}c@{}}Expose or\\ withhold under\\ conflict\end{tabular} & \begin{tabular}[c]{@{}c@{}}Typed\\ \textsc{Conflicting}\\ terminal label\end{tabular} & \begin{tabular}[c]{@{}c@{}}Directional\\ commitment is\\ error target\end{tabular} \\
\midrule
ECON; Tug-of-War; ConFiAct; CARE & Resolve toward a reliable source & \(\checkmark\) & -- & -- & -- \\
\addlinespace[3pt]
NatConfQA; CLUE & Expose conflicting answers & -- & \(\checkmark\) & -- & -- \\
\addlinespace[3pt]
Conflict-aware score & Abstain on unreliable answers & -- & \(\checkmark\) & -- & -- \\
\addlinespace[3pt]
WikiRecentChanges & Measure reasoning--prediction transfer & \(\checkmark\) & -- & -- & -- \\
\addlinespace[3pt]
\textbf{CPO} & Diagnose unauthorized direction & -- & \(\checkmark\) & \(\checkmark\) & \(\checkmark\) \\
\bottomrule
\end{tabular}
}
\end{table*}

\section{Background}\label{sec:background}

\paragraph{Selective prediction and the abstention lineage.}
The idea that a classifier should sometimes refuse to answer goes back
to Chow's rule for character recognition~\citep{chow1957optimum}, which
formalized abstention as the optimal action whenever no class posterior
crosses a confidence threshold. Modern selective classification
inherits this framing: a scoring function ranks inputs by reliability,
and a coverage--risk curve traces what is gained by ceding low-confidence
items~\citep{elyaniv2010foundations,geifman2017selectivenet}. Across
this lineage the implicit licence to abstain is \emph{uncertainty about
which directional answer to give}. Cherry-pick Override breaks that
template: the judge is confident, a directional answer is reachable
from the inputs, and the reason to withhold is not in the judge's
posterior alone but in the structure of the evidence: supporting and
refuting strands coexist on a load-bearing subclaim. A purely
confidence-driven abstention rule cannot see this regime, because
nothing in the judge's own signal flags it.

\paragraph{Issuing a verdict versus being entitled to commit it.}
A second, older distinction is more useful here: producing a judgment
and being authorized to make that judgment the system's public output
are not the same act. The separation is familiar in the philosophy of
assertion (one can hold a belief without being entitled to assert it)
and in accountability framings of automated decision-making, where the
component that \emph{decides} is deliberately distinct from the
component that \emph{commits} on the system's behalf. The separation
principle of \S\ref{sec:method:principle} is the LLM-as-judge
operationalization of this older distinction: the typed panel still
\emph{generates} a directional verdict, but a separate controller
decides whether that verdict is authorized to leave the system, and
\textsc{No-Commit} names the controller-internal state of declining
authorization.

\paragraph{\textsc{Conflicting} is not \textsc{Insufficient}.}
Selective prediction often represents every non-answer through a single
reject option, whereas this task contract distinguishes two reasons for a
non-directional verdict. \textsc{Insufficient} denotes inadequate evidence;
\textsc{Conflicting} preserves material support--refutation conflict.
Atomic-claim decomposition for factuality assessment~\citep{min2023factscore}
treats each atomic claim as a separate truth-bearing unit but does not provide
this four-way terminal-action distinction. The judges studied here
distinguish the two when choosing which non-directional label to use,
returning \textsc{Insufficient} on $69\%$ of gold-\textsc{Insufficient}
cases and \textsc{Conflicting} on $77\%$ of gold-\textsc{Conflicting}
cases on AVeriTeC, but they commit directionally at comparable rates on
both (\S\ref{sec:method:not}). The
selective-prediction literature has not had a reason to mark this
boundary, because its single confidence axis cannot represent it; the
task contract introduced in \S\ref{sec:method:contract} is what makes
the distinction operational by exposing \textsc{Conflicting} as a
first-class verdict that the controller can authorize.

\section{Problem Definition and Diagnostic Protocol}\label{sec:method}
\newtheorem{lemma}{Lemma}


\subsection{Task contract and scope}
\label{sec:method:contract}
We treat fact verification as a per-claim
\emph{commitment-authorization} problem. A \emph{typed verification
contract} is a label set $\mathcal{Y}$ partitioned into
\emph{directional} verdicts $\mathcal{D}=\{\textsc{Supports},
\textsc{Refutes}\}$, which assert a direction about the claim, and
\emph{non-directional} verdicts $\mathcal{N}\subseteq\{\textsc{Conflicting},
\textsc{Insufficient}\}$, which decline to. We write \textsc{Insufficient}
throughout for the label AVeriTeC calls \textsc{Not Enough Evidence} and the
prompt spells \texttt{NOT\_ENOUGH\_INFO}. A contract with
$\mathcal{N}\neq\emptyset$ has already decided that some claims should
not receive a direction, and has given the system the vocabulary to say
so. Our object of study is what systems do with that vocabulary.

We do not claim that any particular claim is absolutely undecidable, nor
that downstream tasks must always avoid a leaning answer. The diagnosis
is contract-relative throughout: it applies exactly when the schema
authorizes a non-directional verdict and the claim's gold label is one
(Table~\ref{tab:scope_boundary}).

\begin{table}[t]
\footnotesize
\centering
\setlength{\tabcolsep}{4pt}
\begin{tabular}{p{0.72\columnwidth}c}
\toprule
Condition & In scope? \\
\midrule
Schema exposes a non-directional verdict ($\mathcal{N}\neq\emptyset$) & yes \\
Gold label is \textsc{Conflicting} (materially mixed) & yes: CPO \\
Gold label is \textsc{Insufficient} (evidence absent) & yes: insuff. \\
Schema has only S/R (forced choice, $\mathcal{N}=\emptyset$) & no \\
Gold label is S or R (accuracy error) & no \\
Deployment policy waives non-directional verdicts & no \\
\bottomrule
\end{tabular}
\caption{Scope conditions. The diagnosis is contract-based, and covers
both non-directional verdicts rather than \textsc{Conflicting} alone.}
\label{tab:scope_boundary}
\end{table}

A \emph{typed LLM panel} $\Pi=\{J_1,J_2,J_3\}$ returns a typed verdict
$\hat{y}\in\mathcal{Y}=\{\textsc{S},\textsc{R},\textsc{I},\textsc{C}\}$
and confidences $p_i\in[0,1]$; we aggregate by majority vote with a
\textsc{C} tie-break and write $\bar{p}=\frac{1}{3}\sum_i p_i$. Outputs
partition into three disjoint categories: \textsc{S}/\textsc{R} are
\emph{directional verdicts}; \textsc{C} and \textsc{I} are
\emph{non-directional verdicts}; \textsc{No-Commit} is a
\emph{controller routing state}, not a verification verdict
(\S\ref{sec:method:principle}).

\subsection{Directional overcommitment and its two instances}
\label{sec:method:cco}
For a claim whose gold label is a non-directional verdict
$g\in\mathcal{N}$, \emph{directional overcommitment} is the event
\[
  \mathrm{Ovc}_g \;:=\;
  \{\mathrm{pred}\in\mathcal{D}\;\wedge\;\mathrm{gold}=g\}.
\]
It is a label error with a specific action structure: the system asserts
a direction on a claim for which the contract authorized it to decline.
The two available values of $g$ give two instances, and they are not
interchangeable.

\paragraph{The conflict instance: Cherry-pick Override (CPO).}
When $g=\textsc{Conflicting}$, the evidence contains material support and
material refutation for the same operative claim, and the system returns
one of them. We call this \textbf{Cherry-pick Override}: the name
describes the selection, which is what cherry-picking denotes. $N_{\mathrm{C}}$
denotes this subset and $\mathrm{CPO}\equiv\mathrm{Ovc}_{\mathrm{C}}$.

\paragraph{The insufficiency instance.}
When $g=\textsc{Insufficient}$, the evidence does not establish the claim
in either direction, and the system supplies one anyway. We write
$\mathrm{Ovc}_{\textsc{I}}$ and report it alongside CPO throughout.

\paragraph{Why the conflict instance is the harder one.}
The two instances differ in what would repair them, and this is why CPO
carries its own name rather than being folded into the general event.
Insufficiency overcommitment is a failure \emph{under missing evidence}:
in principle a better retriever, a larger corpus, or a further hop can
supply what the system lacked, after which a direction may become
authorized. CPO is a failure \emph{under complete evidence}. Both strands
were present, in the context window, and the system still collapsed them.
No amount of additional retrieval addresses it, because nothing was
missing. A pipeline that responds to overcommitment by retrieving harder
will improve the insufficiency instance and leave the conflict instance
untouched.

This paper accordingly centres CPO while measuring both, and
\S\ref{sec:results} shows the measurement was necessary: the two rates
are close, and on one substrate the insufficiency instance is the larger
of the two.

\subsection{What directional overcommitment is not}
\label{sec:method:not}
Both instances are distinct from neighbouring failures
(Table~\ref{tab:failure_taxonomy}): direction error and hallucination
concern \emph{wrong direction} or \emph{commit without evidence};
calibration concerns overconfidence regardless of evidence structure;
generic abstention failure concerns \emph{withholding} on a one-sided
claim the system could have committed correctly. Directional
overcommitment is the subtype of label error in which the erroneous
output is a \emph{direction} and the contract offered an alternative that
was not a direction.

\paragraph{\textsc{Conflicting} and \textsc{Insufficient} are different
task states, but not different failure regimes.}
The two non-directional verdicts encode different things:
\textsc{Insufficient} is the contract's response to inadequate evidence,
\textsc{Conflicting} preserves material support--refutation conflict. Our
judges do discriminate between them \emph{when choosing which one to
use}. On AVeriTeC the panel returns \textsc{Insufficient} on $69\%$
($24/35$) of gold-\textsc{Insufficient} cases and \textsc{Conflicting} on
only $14\%$ ($5/35$) of them; conversely it returns \textsc{Conflicting}
on $77\%$ ($116/150$) of gold-\textsc{Conflicting} cases and
\textsc{Insufficient} on only $4\%$ ($6/150$).

They do not discriminate between them when deciding whether to commit at
all. For the same panel $\mathrm{Ovc}_{\textsc{I}}$ is $17\%$ ($6/35$)
against a CPO rate of $18.7\%$ ($28/150$); the rates are within two
percentage points and their Wilson intervals overlap almost completely
($[0.08,0.33]$ and $[0.13,0.26]$). \S\ref{sec:results} reports the same
contrast for four contemporary judges on two substrates and reaches the
same conclusion, with the insufficiency instance the larger of the two
for three of four judges on VitaminC-Mixed.

We therefore state plainly what this paper does not claim: directional
overcommitment is \emph{not} specific to conflict, and CPO is an instance
of the contract-level failure rather than a mechanism of its own. What
makes the conflict instance worth naming and studying separately is not a
distinct trigger but a distinct remedy structure --- it is the instance
that more evidence cannot fix.

\begin{table}[h]
\footnotesize
\centering
\setlength{\tabcolsep}{4pt}
\begin{tabular}{p{0.26\columnwidth}p{0.60\columnwidth}}
\toprule
Failure & Operational definition \\
\midrule
Direction error              & gold $\in$ \{S,R\}, pred $\in$ \{S,R\}, pred $\neq$ gold \\
\addlinespace[2pt]
Hallucination                & pred is directional, not entailed by any evidence in $E$ \\
\addlinespace[2pt]
Calibration error            & overconfidence on any pred, independent of evidence \\
\addlinespace[2pt]
Abstention failure           & withholds on a one-sided claim it could commit correctly \\
\midrule
\multicolumn{2}{p{0.85\columnwidth}}{\emph{Directional overcommitment} (this paper): pred $\in\mathcal{D}$, gold $\in\mathcal{N}$} \\
\addlinespace[2pt]
\quad\textbf{Cherry-pick Override} & \textbf{gold $=$ \textsc{C}: both strands present, one returned} \\
\addlinespace[2pt]
\quad Insufficiency instance & gold $=$ \textsc{I}: no strand present, one returned \\
\bottomrule
\end{tabular}
\caption{Operational distinctions between directional overcommitment and
neighbouring failures. The categories can co-occur; the table separates
their defining conditions. The two instances share a definition and
differ in whether the missing information could have been retrieved.}
\label{tab:failure_taxonomy}
\end{table}

\subsection{Evaluation sample and evidence source}
\label{sec:method:sample}
The AVeriTeC evaluation set is fixed at $N{=}285$ claims with the
composition in Table~\ref{tab:sample_composition}. Two strata are
\emph{exhaustive} rather than sampled: the development split contains
only 38 \textsc{Conflicting} claims and only 35 \textsc{Not Enough
Evidence} claims, and we take all of both. The
\textsc{Conflicting} subset is then topped up to 150 from the training
split, and the \textsc{Supported} and \textsc{Refuted} strata are drawn
from the development split. $N_{\textsc{I}}{=}35$ is therefore not a
design choice but the size of the entire development
\textsc{Not Enough Evidence} class, which bounds the precision of every
\textsc{Insufficient} contrast we report. The realized sample is
hash-pinned as a source map, so the exact 285 claims are reproducible;
the random seed used to draw the three sampled strata was not recorded
at construction time, so those draws are reproducible from the stored
map but not re-derivable from the seed.

\begin{table}[h]
\footnotesize
\centering
\begin{tabular}{llrrl}
\toprule
Gold label & Split & $n$ & Available & Selection \\
\midrule
\textsc{Conflicting}  & dev   & 38  & 38  & exhaustive \\
\addlinespace[2.5pt]
\textsc{Conflicting}  & train & 112 & 195 & sampled \\
\addlinespace[2.5pt]
\textsc{Insufficient} & dev   & 35  & 35  & exhaustive \\
\addlinespace[2.5pt]
\textsc{Supported}    & dev   & 50  & 122 & sampled \\
\addlinespace[2.5pt]
\textsc{Refuted}      & dev   & 50  & 305 & sampled \\
\midrule
Total & & 285 & & \\
\bottomrule
\end{tabular}
\caption{Composition of the fixed AVeriTeC evaluation sample.}
\label{tab:sample_composition}
\end{table}

\paragraph{Evidence is annotator-curated, not retrieved.}
Each judge receives AVeriTeC's gold question--answer evidence pairs for
the claim, rendered in their stored order; no retrieval system is run.
This is an oracle-retrieval setting, and the choice is deliberate: it
removes the confound in which a directional verdict is correct given a
retrieved record that simply omits the opposing strand. A directional
commitment measured here cannot be attributed to the system never having
seen the other side. The cost is that the rates we report are not
comparable to those of a deployed RAG pipeline, and our external
validity extends only to gold-evidence verification.

\subsection{A same-denominator diagnostic protocol}
\label{sec:method:protocol}
We evaluate commitment authorization with six metrics reported on a
\emph{single denominator} $N$, so that comparisons across controllers
cannot be gamed by shifting the evaluated subset. Let $N_{\mathrm{S/R}}$
and $N_{\mathrm{C}}$ denote the counts of pure-S/R and
gold-\textsc{Conflicting} cases.

\begin{itemize}\setlength\itemsep{1pt}\setlength\parskip{0pt}
  \item $\mathrm{Cov}$ (commit coverage): fraction of $N$ committed to
        $\textsc{Supports}/\textsc{Refutes}$.
  \item $\mathrm{SE}$ (selective error): fraction of committed cases
        that are wrong (subset denominator).
  \item $\mathrm{CPO}_N$: fraction of $N$ that are
        gold-\textsc{Conflicting} yet committed directionally; the
        reviewer-proof rate for cross-controller comparison.
  \item $\mathrm{CPO}_{\mathrm{C}}$: fraction of the
        gold-\textsc{Conflicting} subset ($N_{\mathrm{C}}$) that the
        system commits directionally; used to compare judges restricted
        to the mixed-evidence regime.
  \item $\mathrm{Acc}_{\mathrm{S/R}}$: directional accuracy on the
        pure-evidence subset, with denominator $N_{\mathrm{S/R}}$.
  \item $\mathrm{Rec}_{\mathrm{C}}$: fraction of
        gold-\textsc{Conflicting} claims predicted \textsc{Conflicting},
        with denominator $N_{\mathrm{C}}$.
\end{itemize}

\noindent
Controllers are compared only at \emph{matched coverage}; paired
differences are reported with a \emph{paired bootstrap (5000 resamples)}
on per-case outcomes, without multiple-comparison correction. To test
whether a controller's veto behaviour is \emph{structurally targeted}
rather than merely an equal-sized random promotion, we add a
\emph{random Stage-1 null} that matches the controller's action
mechanically: if the controller promotes $k$ of the baseline's
directional commits to \textsc{Conflicting}, the null draws $2000$
random subsets of the same size $k$ from the baseline's directional
commits and promotes them to \textsc{Conflicting} as well. We report
the controller's empirical one-sided $p$-value against the resulting
distribution on each metric. A small $p$ here is \emph{structural
selectivity} (the controller picked a non-random subset to promote); the
test does not certify magnitude dominance. A non-matched alternative
that demotes commits to \textsc{No-Commit} would leave
$\mathrm{Rec}_{\mathrm{C}}$ mechanically unchanged and is therefore not
the fair null for this comparison.

\subsection{The separation principle}
\label{sec:method:principle}
Standard LLM-as-judge pipelines treat the judge's verdict as the
system's commitment, conflating verdict \emph{generation} with verdict
\emph{authorization}. Under mixed evidence the two come apart: a strong
judge can recognize internal conflict yet still output a direction
because the prompt or label vocabulary requires one. The
\emph{separation principle} states that these decisions should be
implemented by separate components, with an external lightweight
controller mediating between the LLM's typed proposal and the system's
authorized commitment. \textsc{No-Commit} is a controller-internal
routing state, not a verification verdict: it marks the controller
declining to authorize a commitment and leaves the downstream choice
(escalate to a human, re-retrieve, defer) to the deploying application.

\subsection{A minimal two-channel reference probe}
\label{sec:method:controllers}
This probe is included only to instantiate the separation principle in a
transparent way; it is not an optimized controller and should not be
read as a solution to CPO. A \emph{certificate validator} $V$, run
independently of the panel, decomposes the claim into subclaims
$\{s_j\}$ and assigns each an evidence state
$\tau_j\in\{\mathrm{supports},\mathrm{refutes},\mathrm{mixed},
\mathrm{insufficient}\}$~\citep{schlichtkrull2023averitec}. Operationally, a subclaim is
\emph{material} when the validator does not assign it the
\texttt{insufficient} state; we write $\mathrm{material\_mixed}(c, E) :=
\exists\, j\colon \tau_j=\mathrm{mixed} \wedge s_j$ is material. We
treat materiality as an \emph{operational output of the validator}, not
as a ground-truth semantic primitive: the decomposition is imperfect
(Cases~5 and~6 in Appendix~\ref{sec:appendix:cases}), so
\texttt{material\_mixed} functions as a structural veto input rather
than a settled judgment about the evidence.

\paragraph{Probe E (confidence only).}
Authorize a directional commitment only when $\bar{p}\ge\tau$;
otherwise return \textsc{No-Commit}. Non-directional proposals pass
unchanged.

\paragraph{Probe F (two-channel: structural veto $+$ confidence).}
Stage~1 vetoes any directional proposal flagged
\texttt{material\_mixed} (downgrade to \textsc{Conflicting}); Stage~2
applies the confidence threshold as in $E$ (Algorithm~\ref{alg:vgrctc}).
The probe has two deterministic stages, one scalar threshold, and no
learned components.

\begin{algorithm}[t]
\caption{Minimal two-channel authorization probe.}
\label{alg:vgrctc}
\begin{algorithmic}[1]
\REQUIRE typed panel proposal $\hat{y}$; per-judge confidences $\{p_i\}$;
         validator subclaim states $\{\tau_j\}$; threshold $\tau$
\ENSURE  authorized commitment $\tilde{y}\in\mathcal{Y}\cup\{\textsc{No-Commit}\}$
\STATE $\tilde{y}\gets\hat{y}$
\IF{$\tilde{y}\in\{\textsc{S},\textsc{R}\}$ \AND $\mathrm{material\_mixed}(\{\tau_j\})$}
  \STATE $\tilde{y}\gets\textsc{Conflicting}$ \quad{\small // Stage 1: structural veto}
\ENDIF
\IF{$\tilde{y}\in\{\textsc{S},\textsc{R}\}$ \AND $\bar{p}<\tau$}
  \STATE \RETURN \textsc{No-Commit} \quad{\small // Stage 2: confidence gate}
\ENDIF
\RETURN $\tilde{y}$
\end{algorithmic}
\end{algorithm}

\begin{lemma}[Monotonic veto]\label{lem:monotonic}
For any fixed $\tau$, the directional commits of $F$ are a subset of
those of $E$: $\{c : F(c;\tau)\in\{\textsc{S},\textsc{R}\}\} \subseteq
\{c : E(c;\tau)\in\{\textsc{S},\textsc{R}\}\}$.
\end{lemma}
\noindent\emph{Proof sketch.} Stage~2 is identical in $E$ and $F$. $F$
adds Stage~1, which downgrades a strict subset of $E$'s directional
proposals before Stage~2 runs. \rule{1ex}{1ex}

\noindent
By Lemma~\ref{lem:monotonic}, $F$ can never widen the directional
commit set at fixed $\tau$, so any improvement on
$\mathrm{Acc}_{\mathrm{S/R}}$ or $\mathrm{Rec}_{\mathrm{C}}$ must come
from \emph{selectivity} (Stage~1 blocks the right cases), not from
larger coverage.

\paragraph{Implementation.}
Panel: Claude Haiku~4.5 + Claude Sonnet~4.5 + GPT-4o-mini, prompted with
a single shared typed-verdict template. Validator: Claude Haiku~4.5 with
a few-shot certificate prompt; we audit prompt sensitivity in
\S\ref{sec:ablation}. Decoding is greedy per judge; the validator
decomposition is sampled once per claim. Both prompts are reproduced in
full in Appendix~\ref{sec:appendix:prompts}.

\paragraph{The verdict prompt maximally favours \textsc{Conflicting}.}
This is the single most important qualification on every rate in this
paper, and it runs in the paper's favour. The shared template does not
merely list \textsc{Conflicting} as an option. It enumerates four
explicit trigger conditions for it, including cherry-picking by
selective truth and true-core-with-misleading-framing; it carries an
imperative instruction not to collapse a conflicting case to
\textsc{Refutes} even when the misleading aspect dominates; and it
supplies three worked demonstrations, each of which resolves to
\textsc{Conflicting}. A judge that followed the template literally
would over-produce \textsc{Conflicting}, and indeed all judges return
false \textsc{Conflicting} on some pure-evidence claims
(\S\ref{sec:results}).

Every directional-overcommitment rate we report is therefore a
\emph{conservative lower bound} on what a neutrally specified contract
would produce. This is a measured claim rather than an assumption. Under
a neutral specification that names the four labels and says nothing more,
CPO rises for all four contemporary judges, to
$0.593$--$0.644$, with paired intervals excluding zero in every case, while false \textsc{Conflicting} on one-sided evidence
\emph{falls} rather than rises (\S\ref{sec:ablation},
Table~\ref{tab:prompt_ladder}). The instruction is doing substantial work,
and it is not doing it by making the judge label everything conflicting.
The finding is therefore not that models fail to use a label they were
merely offered; it is that they commit directionally at the rates reported
here \emph{after} being told, with worked examples, to prefer the
non-directional verdict.

\paragraph{Retry policy and failure accounting.}
Each request is issued once with at most one retry, triggered only by an
HTTP or transport error and never by an unwanted verdict or an
unparseable response, so retrying cannot select on outcome. Across the
1{,}140 requests of the contemporary AVeriTeC matrix, 21 were retried
after a transport failure (14 HTTP 503 and 7 HTTP 429, all from the
open-weight router) and all 21 then succeeded. Two Claude Sonnet~5
requests produced no usable verdict: one terminal refusal that returned
no text block (a gold-\textsc{Conflicting} case) and one response whose
body was the string \texttt{None} (a gold-\textsc{Insufficient} case).
Both are retained in their fixed denominators and counted as
non-directional, so Claude Sonnet~5 contributes 283 parsed verdicts
while still being scored against $N_{\mathrm{C}}{=}150$ and
$N_{\textsc{I}}{=}35$. This convention can only understate its
overcommitment rates.

Greedy decoding does not make commercial endpoints bit-reproducible, and
we did not issue repeated draws per item. Every interval in this paper is
therefore computed over items, not over runs, and does not bound
run-to-run variation.

\paragraph{Contemporary-model replication.}
To test whether the diagnosis survives model turnover, we freeze four
individual-judge configurations before inspecting their outputs:
GPT-5.6 Terra (\texttt{gpt-5.6-terra}), Claude Sonnet~5
(\texttt{claude-sonnet-5}), gpt-oss-120b
(\texttt{openai/gpt-oss-120b}), and Mistral Small~4
(\texttt{mistralai/mistral-small-2603}). The first two are proprietary;
the latter two have open weights. Each model receives the same four-option
template and the same fixed AVeriTeC ($N{=}285$) and VitaminC-Mixed
($N{=}250$) inputs used above, and is scored individually without panel
aggregation or post-hoc selection. We disable or minimize reasoning and use
deterministic decoding controls where each API supports them. The
open-weight endpoints are pinned to Groq and Mistral, respectively, with
provider fallback disabled. Run manifests record exact checkpoint
identifiers, prompt and source hashes, request settings, and returned
provider identity. We retain parse failures in the fixed denominator and
report Wilson $95\%$ intervals for $\mathrm{CPO}_{\mathrm{C}}$.

\paragraph{Code and data.}
The runners, the analysis scripts and the raw model outputs behind every
table and figure in this paper are at
\url{https://github.com/HrxuAlbert/cherry-pick-override}. Releasing the raw
outputs rather than only the computed rates is deliberate: each table can be
re-derived from the model responses rather than read off our summaries. One
thing is withheld. The calibration round reported in
\S\ref{sec:limits:contract} consists of two people's judgements on twenty
claims, and per-reviewer labels would identify them, so the annotations and
the blinding map are not released; the script that computes the agreement
statistics is, so those numbers can be checked against the code that produced
them.

\section{Diagnostic Analyses of Directional Overcommitment}\label{sec:results}

We organize the design space of common CPO mitigations as an
\emph{intervention ladder} and ask, rung by rung, what each fix
recovers and what residual failure remains. The ladder is a
design-space map, not a leaderboard. Rungs L0--L4 are described without
reference to the two-channel probe, which enters only at L5 to make the
separation principle operational. Table~\ref{tab:ladder_residuals}
summarizes the residual-failure pattern; the table cells are the
sentences each rung's subsection unpacks.

\begin{table}[t]
\centering
\footnotesize
\setlength{\tabcolsep}{3pt}
\caption{\textbf{Residual-failure ladder.} Each tested intervention leaves a
residual; the ladder is diagnostic, not competitive. Design
implications are unpacked in each rung's subsection.}
\label{tab:ladder_residuals}
\begin{tabular}{lp{0.32\columnwidth}p{0.45\columnwidth}}
\toprule
Rung & Added channel & Residual failure \\
\midrule
L0$^{\dagger}$ & 3-option judge                       & by construction directional on \textsc{Conflicting} gold ($>$84\%); its third option means \emph{absence}, not \emph{conflict} \\
\addlinespace[4pt]
L1 & typed \textsc{Conflicting}              & first in-contract rung; structurally avoidable residual remains \\
\addlinespace[4pt]
L1.5 & stronger \textsc{Conflicting} specification & removes 23--57 unauthorized commitments per judge relative to a neutral wording, while producing fewer false \textsc{Conflicting}; a large residual survives the instruction \\
\addlinespace[4pt]
L2 & panel aggregation                  & raises direction-on-conflict in the pre-contract AVeriTeC baseline; typed majority can suppress conflict dissent \\
\addlinespace[4pt]
L3 & confidence threshold                 & calibrated for direction (ECE $0.07$) but cannot operationally separate CPO from correct commits at tested $\tau$ \\
\addlinespace[4pt]
L4 & validator-only filter                & classifier collapses $\mathrm{Acc}_{\mathrm{S/R}}$ ($0.78{\to}0.39$); veto raises $\mathrm{SE}$ over L3 \\
\addlinespace[4pt]
L4.5 & self-decomposition before verdict & lowers CPO for 3 of 4 judges but recovers only about half of the overcommitted claims; on the strongest judge most of the gain is bought by false \textsc{Conflicting} \\
\addlinespace[4pt]
L5 & two-channel reference probe          & magnitude $\Delta$ over L3 modest (CI straddles 0); residual subclaim-mismatch failures \\
\bottomrule
\end{tabular}\\[2pt]
\footnotesize $^{\dagger}$L0 is a \emph{pre-contract} baseline: a 3-option schema does not expose \textsc{Conflicting} and therefore cannot, strictly, measure CPO under our contract (Table~\ref{tab:scope_boundary}). We report L0 to show that without the in-contract vocabulary the failure manifests as a near-mechanical directional rate on gold-\textsc{Conflicting}; the diagnostic ladder begins in-contract at L1.
\end{table}

\subsection{Experimental setup}
\label{sec:results:setup}
We evaluate on two fact-verification substrates. \textbf{AVeriTeC} is a
stratified sample of $N=285$ claims from the AVeriTeC shared
task~\citep{schlichtkrull2023averitec}, whose Conflicting/Cherrypicking
class ($N=150$) supplies naturally occurring mixed evidence.
\textbf{VitaminC-Mixed} is a constructed stress test of $N=250$ claims
built from VitaminC contrastive pairs~\citep{schuster2021vitaminc}, where
mixed evidence is \emph{synthesized} by concatenating supporting and
refuting sources. Because those pairs are Wikipedia revisions taken
before and after an edit, some of the resulting conflicts are plausibly
temporal updates that a time-scoped claim would license resolving; we
therefore report VitaminC-Mixed as behaviour under a constructed
condition rather than as a second prevalence estimate
(\S\ref{sec:limits:data}). The typed panel is Claude Haiku~4.5, Claude
Sonnet~4.5, and GPT-4o-mini; the validator is Claude Haiku~4.5 with the
certificate prompt. All metrics share a single denominator
(Section~\ref{sec:method:protocol}); controllers are compared only at
matched coverage; the paired bootstrap uses 5000 resamples and the
random-veto control uses 2000 seeds. Table~\ref{tab:main_results}
reports the ladder on AVeriTeC; Figure~\ref{fig:random_veto} visualises
the structural-selectivity property of the L5 controller's veto under a
matched-coverage random-veto control.

\paragraph{Contemporary individual-judge replication.}
Table~\ref{tab:current_models} reports the separately frozen replication.
All four models make dataset-defined CPO errors on both substrates. On
AVeriTeC, $\mathrm{CPO}_{\mathrm{C}}$ ranges from $0.213$ to $0.480$; on
VitaminC-Mixed it ranges from $0.240$ to $0.380$. Every corresponding
Wilson interval has a lower bound above zero. These observations show that
the failure is not confined to the historical panel, but they do not imply
that its prevalence has been human-confirmed.

\begin{table*}[t]
\centering
\footnotesize
\setlength{\tabcolsep}{3pt}
\caption{\textbf{Contemporary individual-judge replication.} Overcommitment
is scored against each dataset's four-way labels. Parenthesized ranges are
Wilson $95\%$ intervals. $\mathrm{Ovc}_{\textsc{I}}$ is the directional
rate on gold-\textsc{Insufficient} cases, the control for whether the
failure is specific to conflict. Substrates are stacked as row blocks rather than
placed side by side: at ten columns the table did not fit the text
width. All
VitaminC-Mixed outputs and all AVeriTeC outputs except two from Claude
Sonnet~5 parsed successfully; parse failures remain in the fixed
denominator and are counted as non-directional.}
\label{tab:current_models}
\begin{tabular}{llcccc}
\toprule
Model & Access & $\mathrm{CPO}_{\mathrm{C}}\downarrow$ & $\mathrm{Ovc}_{\textsc{I}}\downarrow$ & $\mathrm{Rec}_{\mathrm{C}}\uparrow$ & $\mathrm{Acc}_{\mathrm{S/R}}\uparrow$ \\
\midrule
\multicolumn{6}{l}{\emph{AVeriTeC} ($N_{\mathrm{C}}{=}150$, $N_{\textsc{I}}{=}35$, $N_{\mathrm{S/R}}{=}100$)} \\
GPT-5.6 Terra   & proprietary & $0.293\;(0.226,0.371)$ & $0.086\;(0.030,0.224)$ & 0.613 & 0.800 \\
Claude Sonnet~5 & proprietary & $0.300\;(0.232,0.378)$ & $0.200\;(0.100,0.359)$ & 0.633 & 0.870 \\
gpt-oss-120b    & open weight & $0.480\;(0.402,0.559)$ & $0.286\;(0.163,0.451)$ & 0.373 & 0.790 \\
Mistral Small~4 & open weight & $0.213\;(0.155,0.286)$ & $0.200\;(0.100,0.359)$ & 0.720 & 0.680 \\
\midrule
\multicolumn{6}{l}{\emph{VitaminC-Mixed} ($N_{\mathrm{C}}{=}N_{\mathrm{S/R}}{=}100$, $N_{\textsc{I}}{=}50$)} \\
GPT-5.6 Terra   & proprietary & $0.380\;(0.291,0.478)$ & $0.380\;(0.259,0.518)$ & 0.570 & 0.720 \\
Claude Sonnet~5 & proprietary & $0.380\;(0.291,0.478)$ & $0.420\;(0.294,0.558)$ & 0.610 & 0.770 \\
gpt-oss-120b    & open weight & $0.240\;(0.167,0.332)$ & $0.440\;(0.312,0.577)$ & 0.750 & 0.790 \\
Mistral Small~4 & open weight & $0.360\;(0.273,0.458)$ & $0.480\;(0.348,0.615)$ & 0.620 & 0.740 \\
\bottomrule
\end{tabular}
\end{table*}

\paragraph{The failure is not specific to conflict.}
The $\mathrm{Ovc}_{\textsc{I}}$ column tests whether directional
overcommitment is a response to \emph{conflicting} evidence in
particular or to the presence of \emph{any} non-directional verdict in
the contract. It is not conflict-specific. On AVeriTeC the gap
$\mathrm{CPO}_{\mathrm{C}}-\mathrm{Ovc}_{\textsc{I}}$ is positive for
all four judges but ranges from $+0.013$ to $+0.208$, and the
$N_{\textsc{I}}{=}35$ intervals are too wide to separate the two rates
for Claude Sonnet~5 or Mistral Small~4. On VitaminC-Mixed the gap
reverses: three of the four judges overcommit \emph{more} on
gold-\textsc{Insufficient} than on gold-\textsc{Conflicting}
($-0.040$, $-0.200$ and $-0.120$; GPT-5.6 Terra is exactly $0.000$).
The historical panel behaves the same way, at $0.187$ on
gold-\textsc{Conflicting} against $0.171$ on gold-\textsc{Insufficient}
(\S\ref{sec:method:not}).

The two subsets are disjoint sets of claims, so this is an unpaired
comparison of proportions and not the per-case paired contrast used
elsewhere in this section; with $N_{\textsc{I}}$ of 35 and 50 it is also
the least precise quantity we report. We therefore draw only the
conclusion the design supports: the evidence does not license a
conflict-specific reading, and any account of the failure must also
explain the comparable rate at which these judges decline
\textsc{Insufficient}.

We read this as generalizing the diagnosis rather than undermining it.
What the judges are doing is not reacting to conflict; it is declining
the non-directional half of their own output vocabulary, whichever half
the contract offers. The distinction that survives is the one drawn in
\S\ref{sec:method:cco}, and it is about remedy rather than trigger: the
insufficiency instance is a failure under missing evidence and is
therefore addressable by retrieving more, while CPO occurs with both
strands already in context and is not. A system that responds to the
aggregate rate by strengthening retrieval would move the
$\mathrm{Ovc}_{\textsc{I}}$ column and leave the CPO column where it is.

The ordering is not stable across substrates: gpt-oss-120b has the highest
AVeriTeC CPO rate but the lowest VitaminC-Mixed rate, whereas Mistral
Small~4 moves from the lowest to the second-highest. Low CPO also need not
mean uniformly better verification: on AVeriTeC, Mistral Small~4's lower
CPO coincides with the lowest pure-S/R accuracy ($0.680$) and the highest
false-conflict rate on pure S/R cases ($0.180$). We therefore report
individual models rather than pooling by access regime, and draw no general
proprietary-versus-open-weight conclusion from two models per group. These
runs diagnose benchmark-contract behaviour only; an independent human review,
which we designed and did not complete
(Section~\ref{sec:limits:contract}), is required before interpreting the
rates as semantically confirmed overcommitment.

\begin{table*}[t]
\centering
\footnotesize
\setlength{\tabcolsep}{4pt}
\caption{Intervention ladder on AVeriTeC mixed-evidence claims
($N{=}285$, stratified; $N_{\mathrm{S/R}}{=}100$,
$N_{\mathrm{C}}{=}150$). Metrics defined in
Section~\ref{sec:method:protocol}; all share denominators across systems.
$\mathrm{CPO}_N$ is the same-denominator rate over $N$ (used for
cross-system comparison); $\mathrm{CPO}_{\mathrm{C}}$ is the conditional
rate on the gold-\textsc{Conflicting} subset. Paired-bootstrap CIs for
the key differences are reported in the paragraph below the table.}
\label{tab:main_results}
\begin{tabular}{lrrrrrr}
\toprule
System & $\mathrm{Cov}\,\uparrow$ & $\mathrm{SE}\,\downarrow$ & $\mathrm{CPO}_N\,\downarrow$ & $\mathrm{CPO}_{\mathrm{C}}\,\downarrow$ & $\mathrm{Acc}_{\mathrm{S/R}}\,\uparrow$ & $\mathrm{Rec}_{\mathrm{C}}\,\uparrow$ \\
\midrule
\multicolumn{7}{l}{\emph{L0: three-option direct judge (pre-contract; reported for illustration)$^{\dagger}$}} \\
Haiku 3-opt          & 0.800 & 0.632 & (0.442) & (0.840) & 0.840 & 0.000 \\
Panel 3-opt (majority) & 0.846 & 0.606 & (0.467) & (\textbf{0.887}) & 0.950 & 0.000 \\
\multicolumn{7}{l}{\emph{L1: typed four-option judge}} \\
Panel + typed (majority) & 0.396 & 0.310 & 0.098 & 0.187 & 0.780 & 0.773 \\
\multicolumn{7}{l}{\emph{L2: alternative panel aggregation}} \\
Conflict-if-any panel & 0.305 & 0.230 & 0.053 & 0.100 & 0.670 & 0.893 \\
\multicolumn{7}{l}{\emph{L3: confidence-threshold selective}} \\
E $\tau$=0.90        & 0.288 & 0.207 & 0.056 & 0.107 & 0.650 & 0.773 \\
\multicolumn{7}{l}{\emph{L4: validator-only filtering}} \\
Validator-veto       & 0.295 & 0.226 & 0.053 & 0.100 & 0.650 & 0.860 \\
\multicolumn{7}{l}{\emph{L5: two-channel reference probe}} \\
F $\tau$=0.85        & 0.281 & 0.200 & 0.046 & 0.087 & 0.640 & 0.860 \\
F $\tau$=0.90        & 0.232 & 0.152 & 0.032 & 0.060 & 0.560 & 0.860 \\
\bottomrule
\end{tabular}\\[2pt]
\footnotesize $^{\dagger}$Parentheses on L0's $\mathrm{CPO}_N$ and $\mathrm{CPO}_{\mathrm{C}}$ entries indicate descriptive reporting only: the 3-option schema does not expose \textsc{Conflicting} as a verdict, so these cells fall outside our contract for CPO measurement (Table~\ref{tab:scope_boundary}); they quantify the precondition failure that motivates L1.
\end{table*}

\paragraph{Cross-dataset replication on VitaminC-Mixed.} The full
per-system VitaminC-Mixed table mirroring Table~\ref{tab:main_results} is
provided in Appendix~\ref{sec:appendix:vitaminc}
(Table~\ref{tab:vitaminc_results}). Two findings transfer: (a) the L0
panel-amplification effect does \emph{not} replicate
($\Delta\,\mathrm{CPO}_{\mathrm{C}}=0.000$, $95\%$ CI $[-0.060,+0.060]$;
$N_{\mathrm{C}}=100$), so we scope the amplification claim to AVeriTeC;
(b) the L3$\to$L5 operating-point pattern at $\tau{=}0.90$ matches the
AVeriTeC main operating point (L5 $\mathrm{SE}=0.388$ vs L3
$0.416$; $\mathrm{CPO}_N=0.084$ vs $0.104$) and ties in the extreme
low-coverage regime at $\tau{=}0.95$.

\paragraph{Paired-bootstrap CIs.}
On per-case paired differences ($5000$ resamples, seed $0$):
\emph{(i)} L0 panel amplification on AVeriTeC ($N_{\mathrm{C}}=150$),
$\Delta\,\mathrm{CPO}_{\mathrm{C}}=+0.047$, CI $[+0.013,+0.080]$,
separated; \emph{(ii)} L0$\to$L1 vocabulary fix,
$\Delta\,\mathrm{CPO}_{\mathrm{C}}=-0.653$, CI $[-0.733,-0.573]$,
separated; \emph{(iii)} validator-as-classifier collapse
($N_{\mathrm{S/R}}=100$), $\Delta\,\mathrm{Acc}_{\mathrm{S/R}}=-0.390$,
CI $[-0.490,-0.290]$, separated; \emph{(iv)} L5 vs L3 at matched
coverage $\approx0.28$, $\Delta\,\mathrm{SE}=-0.007$, CI
$[-0.084,+0.066]$ and $\Delta\,\mathrm{CPO}_N=-0.010$, CI
$[-0.035,+0.011]$; both straddle $0$. The robust L3$\to$L5 finding is
structural selectivity (Figure~\ref{fig:random_veto}), not aggregate
magnitude.

\begin{figure*}[t]
\centering
\includegraphics[width=\textwidth]{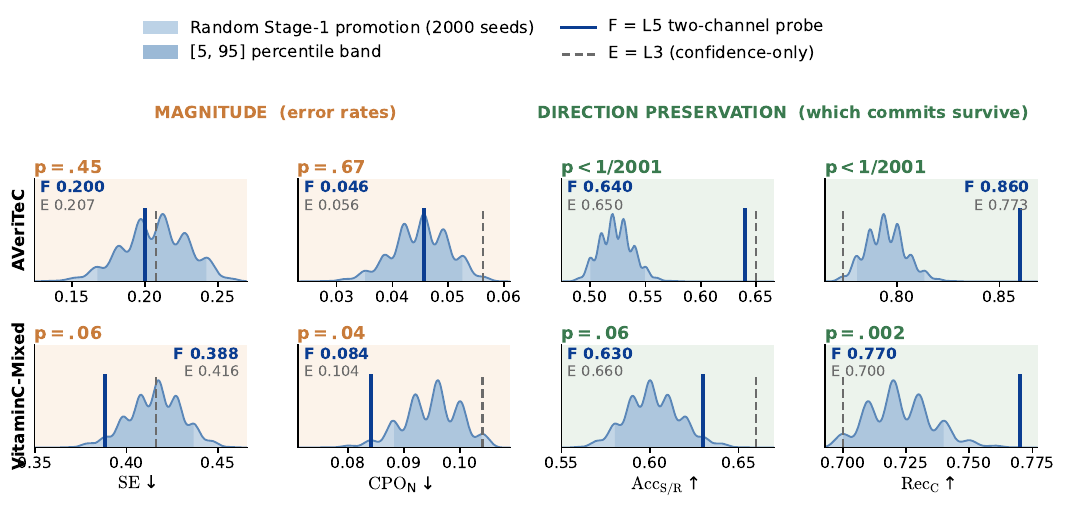}
\caption{\textbf{Matched random-promotion control.} Rows are the two
substrates, AVeriTeC ($N{=}285$) and VitaminC-Mixed ($N{=}250$). Stage~1
promotes $k$ of L3's directional commits to \textsc{Conflicting}
($k{=}16$ on AVeriTeC, $k{=}10$ on VitaminC-Mixed); the null promotes
$2000$ random size-$k$ subsets. The shaded region covers the 5th--95th
percentiles of the null. Each panel reports $p$, the fraction of those
$2000$ draws that match or beat the probe on that measure, so a small $p$
means the probe's choice of which commits to promote was not
random. Columns are the error-rate measures $\mathrm{SE}$ (selective
error) and $\mathrm{CPO}_N$ (same-denominator CPO rate), then the
direction-preservation measures $\mathrm{Acc}_{\mathrm{S/R}}$ (pure-S/R
accuracy) and $\mathrm{Rec}_{\mathrm{C}}$ (conflict recall); the arrow
on each axis gives the preferred direction. On AVeriTeC the error-rate
measures are indistinguishable from the null while both
direction-preservation measures reach $p<1/2001$; on VitaminC-Mixed only
$\mathrm{Rec}_{\mathrm{C}}$ stays clearly separated and the rest sit
between $0.04$ and $0.07$. The two direction-preservation measures are
algebraically coupled under fixed $k$ and represent one selection event,
not independent tests.}
\label{fig:random_veto}
\end{figure*}

\subsection{L0\,$\to$\,L1: the in-contract starting point}
\label{sec:results:l0l1}
L0 is a \emph{pre-contract} baseline: a three-option judge
(\textsc{S}/\textsc{R}/NEI) cannot, by construction, return
\textsc{Conflicting}, so any \textsc{Conflicting} gold case it sees
becomes a forced directional commitment (rate $>$$0.84$ on AVeriTeC).
We report L0 only to motivate L1 as the first \emph{in-contract} rung of
the ladder. Because our contract scopes CPO to schemas that expose
\textsc{Conflicting} (Table~\ref{tab:scope_boundary}), we do not score L0
as a CPO controller. The
single largest reduction in the ladder occurs at this contract entry
point (L1 $\mathrm{CPO}_N=0.098$ on AVeriTeC; $0.120$ on
VitaminC-Mixed). A structurally avoidable residual remains, which the
remaining rungs L2--L5 are designed to interrogate.

\subsection{L2: panel voting (a scoped negative finding)}
\label{sec:results:l2}
Under the three-option schema on AVeriTeC, three-judge majority voting
amplifies the direction-on-conflict rate (panel $0.887$ vs single
$0.840$, $+0.047$, CI $[+0.013,+0.080]$); on VitaminC-Mixed the effect
is absent (CI $[-0.060,+0.060]$). To explain the asymmetry, we inspect
per-judge vote distributions on the gold-\textsc{Conflicting} subsets
(Table~\ref{tab:amplification_anatomy}). Under 3-opt roughly two-thirds
of conflicting cases have \emph{all three} judges voting directionally on
both datasets (AVeriTeC 71\%, VitaminC 66\%); amplification is
dominated by shared directional bias rather than by the majority
overriding minority dissent. Under the in-contract 4-opt typed schema,
we then inspect the panel agreement structure
within the $N{=}27$ AVeriTeC CPO commits themselves\footnote{One of
the $28$ CPO commits used in the L3 confidence analysis below had a
per-judge vote count $\neq 3$ (a parse failure on one judge) and is
omitted here, since the 3-0/2-1 anatomy requires a complete panel.} and find that the
mechanism is split nearly evenly between two failure modes: $52\%$
($14/27$, Wilson $[0.34,0.69]$) are \emph{unanimous on the committed
direction} (shared bias across judges) and $48\%$ ($13/27$,
$[0.31,0.66]$) are
\emph{2-vs-1 with the dissenting judge voting non-directionally}
(majority suppression of a single judge who flagged conflict). Almost
half of all CPO is therefore an aggregation artifact in which a single
judge produced the conflict-preserving verdict and was outvoted, indicating that
typed-vocabulary alone is insufficient when paired with majority
aggregation. An alternative aggregation
rule that commits only when no judge votes \textsc{Conflicting} raises
$\mathrm{Rec}_{\mathrm{C}}$ to $0.893$ and lowers $\mathrm{CPO}_N$ to
$0.053$, at the cost of $\mathrm{Acc}_{\mathrm{S/R}}$ ($0.670$ vs typed
direct $0.780$) and a higher $\mathrm{SE}$ than L3
(Table~\ref{tab:main_results}). Aggregation choices expose different
conflict-recall--coverage trade-offs rather than producing a single dominant
intervention.

\subsection{L3: confidence-threshold selection}
\label{sec:results:l3}
Committing only when $\bar{p}\ge\tau$ traces a clean risk--coverage
curve (E $\tau{=}0.90$: $\mathrm{Cov}\,0.288$, $\mathrm{SE}\,0.207$,
$\mathrm{CPO}_N\,0.056$). The relevant question for our diagnosis is
not whether the curve exists but whether it can disentangle CPO from
correct directional commits. It cannot, in any operationally useful
sense. The panel is in fact well-calibrated for directional
prediction on the pure-S/R subset (ECE $=0.07$ on $N_{\mathrm{S/R}}{=}100$),
so CPO is not the residual of pure-S/R miscalibration. CPO commits
($n{=}28$) have mean $\bar{p}{=}0.902$ versus $0.933$ for correct
directional commits ($n{=}78$ on $N_{\mathrm{S/R}}{=}100$); the two confidence distributions are
statistically distinguishable (Mann--Whitney two-sided $p\!\approx\!0.0002$)
but the mean separation is $0.031$, far below what a usable threshold
needs: $86\%$ of CPO commits sit at $\bar{p}\ge0.85$ ($24/28$, Wilson
$[0.69,0.94]$), $57\%$ at $\bar{p}\ge0.90$ ($16/28$, $[0.39,0.74]$),
and $7\%$ at $\bar{p}\ge0.95$ ($2/28$, $[0.02,0.23]$). No threshold
within the tested range filters CPO without also filtering correct
directional commits, because confidence measures decision strength
rather than evidence structure. The face-mask example of
Section~\ref{sec:intro} (panel \textsc{Refutes} at $\bar{p}{=}0.96$) is
exactly this high-confidence pass-through.

\subsection{L4: validator-only filtering}
\label{sec:results:l4}
If the gap left by L3 is insensitivity to structural conflict, the
natural fix is a deterministic evidence-state validator. As a
\emph{primary classifier} it collapses pure-S/R accuracy from $0.78$
to $0.39$ on AVeriTeC: the deterministic rule cannot weight subclaim
importance against the main directional question. As a \emph{veto
channel} on top of typed proposals it preserves $\mathrm{Acc}_{\mathrm{S/R}}$
($0.650$) and reaches $\mathrm{CPO}_N\,0.053$, but its $\mathrm{SE}$
($0.226$) is higher than L3's ($0.207$). The validator catches
structural conflicts L3 misses while also firing on cases the judge
settles correctly.

\subsection{L5: two-channel external authorization}
\label{sec:results:l5}
L3 and L4 leave non-substitutable residuals: one cannot see structural
conflict, while the other is too brittle alone. The top rung uses each
channel for the role it is good at: the validator as a structural-evidence
veto, confidence as a decision-strength gate. At matched coverage
${\approx}\,0.28$ on AVeriTeC, $F$ $\tau$=0.85 reaches
$\mathrm{SE}\,0.200$ and $\mathrm{CPO}_N\,0.046$ vs L3's $0.207$ and
$0.056$; magnitude improvements are modest and their bootstrap CIs
straddle zero (above). The robust finding is \emph{structural
selectivity}: under the apples-to-apples random Stage-1 null, no random
sample matches L5 on $\mathrm{Acc}_{\mathrm{S/R}}$ (empirical
$p<1/2001$). Because each of the $k$ promotions either decrements
$\mathrm{Acc}_{\mathrm{S/R}}$ (if the promoted case has gold S/R) or
increments $\mathrm{Rec}_{\mathrm{C}}$ (if the gold is
\textsc{Conflicting} and the panel committed directionally), the two
metrics are algebraically coupled under fixed $k$, and we therefore
treat the joint result on both axes as a single selection event
expressed in two ways rather than as two independent tests. The veto
selects \emph{which} cases to block in a structurally non-arbitrary way;
the null does not certify magnitude dominance, only non-arbitrariness. By
Lemma~\ref{lem:monotonic}, $F$ is at-least-as-restrictive as $E$ at
any fixed $\tau$, so selectivity, not extra coverage, is the only
available source of improvement. The same operating-point pattern
replicates on VitaminC-Mixed at $\tau$=0.90 (L5
$\mathrm{SE}\,0.388$, $\mathrm{CPO}_N\,0.084$ vs L3's $0.416$,
$0.104$), and the random-veto null again gives empirical
$p<1/2001$; at $\tau$=0.95 L5 and L3 are tied at near-zero
$\mathrm{CPO}_N\,(0.024)$. Taken together, no single channel L0--L4
simultaneously preserves pure-S/R accuracy and conflict recall at the
matched operating points; L5 is the minimal honest top-rung instance,
exhibiting structural selectivity but not magnitude dominance.



\section{Is It Recognition or Commitment?}\label{sec:recognition}

\begin{figure*}[t]
\centering
\includegraphics[width=\textwidth]{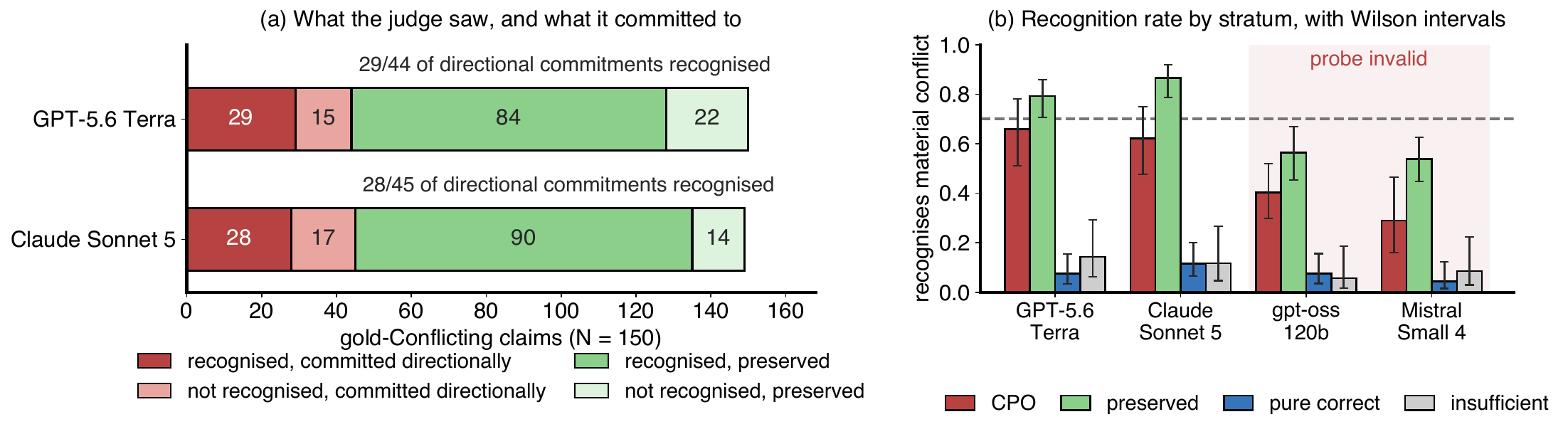}
\caption{\textbf{Recognition and commitment come apart.} \textbf{(a)} The
gold-\textsc{Conflicting} subset ($N_{\mathrm{C}}{=}150$) partitioned by what
the fresh-context probe answered and what the judge actually committed to,
for the two judges whose probe satisfies the validity conditions. The dark red
block is the object of this paper: the judge identified material support and
material refutation, quoting a verbatim span for each, and returned one
direction as its verdict. \textbf{(b)} Probe recognition rate by stratum for
all four judges, with Wilson $95\%$ intervals. The dashed line is the
pre-specified validity condition requiring recognition on at least $0.70$ of
the cases a judge itself returned \textsc{Conflicting} on; gpt-oss-120b and
Mistral Small~4 fall below it, so their readouts are marked invalid (shaded)
and excluded from the claim in \S\ref{sec:recognition:dissociation}. The low
\emph{insufficient} bars show that the probe separates absent evidence from
two-sided evidence for all four judges, including the two whose readouts are
excluded.}
\label{fig:recognition}
\end{figure*}

The results so far establish a rate: judges assert a direction on claims whose
authorized verdict is non-directional. A rate does not say why. Two
explanations are compatible with everything reported above. The judge may
never have represented the evidence as two-sided, in which case the failure is
ordinary evidence misreading and the term \emph{override} is unwarranted. Or
it may have represented the conflict and issued a direction regardless, in
which case the recognition never reached the terminal action. This section
separates the two.

\subsection{A fresh-context recognition probe}
\label{sec:recognition:probe}
For every case in the fixed AVeriTeC sample, each judge is queried a second
time in a context that contains the claim and the same evidence bytes, and
nothing else: no verdict, no confidence, no prior response, no label
vocabulary, and no mention of verification as a task. The question is about
evidence structure alone --- whether the evidence contains material support
\emph{and} material refutation --- and the schema requires a verbatim span for
each direction claimed, so a direction asserted without a quotable span is
recorded as absent. The full probe prompt is in
Appendix~\ref{sec:appendix:prompts}.

Strata are defined by the \emph{parent verdict run}, never by the probe's own
answers, so the probe cannot select its own denominator. Four strata matter
here: the CPO cases (gold \textsc{Conflicting}, verdict directional), the
preserved cases (gold \textsc{Conflicting}, verdict \textsc{Conflicting}), the
pure-evidence cases the judge got right (gold \textsc{Supports} or
\textsc{Refutes}, verdict matching), and the gold-\textsc{Insufficient} cases.
The probe protocol, its estimands, and the validity conditions below were
frozen before the first probe call was issued.

\subsection{Probe validity}
\label{sec:recognition:validity}
A probe of this kind fails in two directions, and both were pre-specified as
disqualifying. It fails if it does not detect conflict where the judge itself
reported conflict, and it fails if it answers yes on one-sided evidence. We
therefore required, per model, that the yes-rate on preserved cases reach
$0.70$, that the yes-rate on correctly committed pure-evidence cases stay
below $0.50$, and that the bootstrap interval for the difference between the
CPO and pure-correct rates exclude zero.

Two of the four judges meet all three conditions and two do not
(Table~\ref{tab:e1_rates}, Figure~\ref{fig:recognition}b). Both failures are on the first condition: on cases
where gpt-oss-120b and Mistral Small~4 themselves returned
\textsc{Conflicting}, the probe recognizes two-sidedness on only $0.564$ and
$0.538$ of cases. For those two models the instrument is not reliably
measuring recognition, and we exclude their readouts from the claim below
rather than reinterpret them. Their rates are shown for completeness and are
descriptive only. This is a limit of the probe, not evidence that these models
lack the dissociation; the question is simply unresolved for them.

\begin{table*}[t]
\centering
\footnotesize
\setlength{\tabcolsep}{4pt}
\caption{Recognition rates by stratum on AVeriTeC, with the pre-specified
validity conditions. $R_{\mathrm{CPO}}$ is the fraction of a judge's own CPO
cases on which it identifies the evidence as materially two-sided when asked
in a fresh context. $\Delta$ is $R_{\mathrm{CPO}}-R_{\mathrm{pure}}$ with a
case-clustered bootstrap interval ($10{,}000$ resamples). The two strata are
disjoint sets of claims, so the comparison is unpaired. Three probe calls did
not return a usable answer for Claude Sonnet~5 and three for Mistral
Small~4, so those judges' strata are one case short of the corresponding
counts in Table~\ref{tab:current_models}; for Mistral Small~4 the missing
case falls in the CPO stratum, which is why its denominator here is $31$
rather than $32$.}
\label{tab:e1_rates}
\begin{tabular}{lccccc}
\toprule
Judge & $R_{\mathrm{CPO}}$ & $R_{\mathrm{pres}}$ & $R_{\mathrm{pure}}$ & $\Delta$ (95\% CI) & Valid \\
\midrule
GPT-5.6 Terra   & $29/44=0.659$ & 0.792 & 0.075 & $+0.584\;[+0.428,+0.730]$ & yes \\
Claude Sonnet~5 & $28/45=0.622$ & 0.865 & 0.116 & $+0.506\;[+0.343,+0.661]$ & yes \\
gpt-oss-120b    & $29/72=0.403$ & \textbf{0.564} & 0.076 & $+0.327\;[+0.197,+0.456]$ & no \\
Mistral Small~4 & $\;\,9/31=0.290$ & \textbf{0.538} & 0.044 & $+0.246\;[+0.083,+0.423]$ & no \\
\bottomrule
\end{tabular}
\end{table*}

\subsection{Recognition and commitment come apart}
\label{sec:recognition:dissociation}
For the two judges where the probe is valid, recognition of the conflict is
present on most of the cases where the judge nonetheless committed to a
direction. GPT-5.6 Terra identifies the evidence as materially two-sided on
$29$ of its $44$ CPO cases ($0.659$); Claude Sonnet~5 does so on $28$ of $45$
($0.622$). On cases where the same judge committed directionally and the
dataset agrees with the direction, the same probe answers yes on $0.075$ and
$0.116$ of cases. The differences are $+0.584$ and $+0.506$, with bootstrap
intervals well clear of zero (Table~\ref{tab:e1_rates}).

Figure~\ref{fig:recognition}a gives the joint distribution on the
gold-\textsc{Conflicting} subset. Every one of the $29$ and $28$ cases in the
recognized-and-committed cell carried verbatim spans from both directions, so
in each of those cases the judge could quote the supporting evidence and the
refuting evidence and still returned one of them as its verdict.

We read this as a dissociation between recognition and commitment for these
two judges: the information needed to withhold the direction was available to
the model, elicitable from it, and did not change the terminal output. This
is what licenses describing the failure as an override rather than as a
misreading, and the scope of that description is exactly these two judges on
this substrate.

One qualification bounds the claim and cannot be removed by more data of this
kind. The probe measures \emph{elicited} recognition, under a question that
directs attention to evidence structure. It does not establish that the same
representation was active while the verdict was being generated. What the
probe shows is that the capacity was present and cheap to surface, not that it
was consulted.

\paragraph{The probe separates the two non-directional states that the
verdicts do not.}
The recognition rate on gold-\textsc{Insufficient} cases is low for every
judge --- $0.143$, $0.118$, $0.057$ and $0.086$ --- so the probe distinguishes
absent evidence from two-sided evidence cleanly, and does so for the two
models whose probe fails the preserved-case condition as well. Set beside the
overcommitment rates in \S\ref{sec:results}, where directional commitment on
gold-\textsc{Insufficient} cases runs at rates comparable to CPO, this
sharpens the reading of that result. The judges are not confusing the two
evidence states. They distinguish them, and commit directionally on both
anyway.

\subsection{The judge overrides evidence it has itself written down}
\label{sec:recognition:counterevidence}
The probe above asks the judge a question it was not otherwise asked. A stricter
test uses only what the judge produced on its own. In the two-stage condition of
\S\ref{sec:ablation}, the judge first lists the evidence items supporting the
claim and those refuting it, and then issues a verdict with that list in front of
it. Where it commits directionally on a gold-\textsc{Conflicting} claim, we can
ask whether its own list already contained items pointing the other way.

It usually did. Of the directional commitments GPT-5.6 Terra made in that
condition, $18/27$ ($0.667$) were made on claims for which it had itself
enumerated opposing items; for Claude Sonnet~5, $23/33$ ($0.697$). The two
open-weight judges are lower, at $27/58$ ($0.466$) and $13/37$ ($0.351$). The
median number of opposing items enumerated and then overridden is two in every
case, so these are not single stray mentions.

This measure needs no second query and no claim about anything the judge did not
write. It also converges with the probe: for the two judges whose probe is
valid, the fresh-context recognition rate and the self-enumeration rate differ by
$0.008$ and $0.075$, although the two instruments share neither a context nor a
question format. Agreement of that closeness between a probe the judge did not
expect and a list the judge volunteered is the strongest support we have that
$R_{\mathrm{CPO}}$ is measuring the intended quantity.

The two measures are not interchangeable. Self-enumeration is observed under an
intervention that asks for a decomposition, so it describes the judge's behaviour
in that condition rather than in the unmodified pipeline, and it counts an item
as opposing if the judge filed it that way rather than by an independent
standard. What it establishes is narrower and harder to dispute than the probe:
on most of these claims the judge wrote down the counterevidence, in its own
words, in the same exchange in which it ruled against it.

\subsection{What a token-level probe could not measure}
\label{sec:recognition:logprob}
A natural objection to the paragraphs above is that a second query is
unnecessary: the verdict token's own probability distribution should reveal
whether \textsc{Conflicting} was competitive at the moment of commitment. We
tested this and report it as a negative methodological result, because the
measurement does not survive contact with how these models emit verdicts.

The first attempt constrained the judge to answer with a single letter, so
that the first emitted token would be the verdict. At that position the
competing candidates are not other verdicts but surface forms: after the
emitted letter, the top alternatives were punctuation and continuation
fragments, while the other three verdict letters sat between $-24$ and $-30$
in log probability. The emitted letter was also not always the first token.

The second attempt removed that problem by reverting to the unmodified
four-option verdict prompt and locating the token that begins the verdict
value inside the model's own required JSON, where the competing candidates
are the verdict words themselves. The instrument then works as intended and
is uninformative for a structural reason. Across $279$ parsed responses the
probability on the emitted verdict has median $1.000$ and minimum $0.99998$;
on all $57$ CPO cases in that run, \textsc{Conflicting} does not appear among
the twenty most likely tokens at all; and \textsc{Conflicting} ranks second in
one case out of $279$. The verdict field follows a reasoning field in the
required schema, so by the time the model writes the verdict it has already
committed in prose, and the verdict token transcribes that commitment rather
than deciding it.

The consequence is a scope limit on black-box uncertainty methods for this
failure. Token-level verdict probabilities cannot probe the commitment
decision in a judge that reasons in text before emitting the verdict, and the
absence of probability mass on the authorized label is not evidence that the
label was never considered. Our evidence for the dissociation therefore rests
on the behavioural probe above and on the intervention in
\S\ref{sec:ablation}.

\subsection{The rates depend on the serving route}
\label{sec:recognition:route}
Because the logprob run used the unmodified verdict prompt, it also produced a
second measurement of the verdict itself for one model, on a different
serving route. Logprobs are unavailable on the route the main matrix pins for
gpt-oss-120b, so that run was served elsewhere with the same model
identifier, the same prompt bytes, the same cases, and greedy decoding.

The two routes return the same verdict on $217$ of $279$ comparable cases
($0.778$), and differ on $0.222$ of them. On the $147$ gold-\textsc{Conflicting}
cases both routes parsed, CPO is $71/147=0.483$ (Wilson $[0.404,0.563]$) on
the pinned route and $57/147=0.388$ ($[0.313,0.468]$) on the alternative. The
intervals overlap, so we do not claim the two rates differ; the point
estimates differ by $9.5$ percentage points, which is larger than several
differences the paper treats as worth reporting between models. The
disagreement is also directional: $28$ cases move from a directional verdict
on the pinned route to \textsc{Conflicting} on the alternative, against $5$
moving the other way.

We report this as a bound on the reproducibility of open-weight results rather
than as a finding about either route. A serving stack is not part of a model's
identity as usually cited, and the observation here is that two stacks serving
one checkpoint under one prompt disagree on roughly a fifth of terminal
verdicts. Any open-weight rate in this paper should be read with that in mind,
and we do not know which route a reader attempting replication would receive.
The measurement covers one model on two routes, so it establishes that the
sensitivity exists without estimating its typical size.

\section{Channel-Complementarity Tests}\label{sec:ablation}

Each test below asks a channel-complementarity question: \emph{does
adding channel $X$ recover something channel $Y$ cannot, or does either
alone suffice?} The goal is diagnosis rather than defense of a controller.
Table~\ref{tab:ablation} reports the matched-coverage numbers on
AVeriTeC.

\begin{table*}[t]
\centering
\footnotesize
\setlength{\tabcolsep}{4pt}
\caption{Channel-complementarity diagnostic on AVeriTeC ($N{=}285$), at
comparable coverage. This is a diagnostic of channel complementarity,
not a method ablation proving superiority.
$^{\ddagger}$The row added for exact coverage matching against Stage~2
alone; see the threshold-sensitivity paragraph.}
\label{tab:ablation}
\begin{tabular}{lrrrr}
\toprule
Controller & $\mathrm{Cov}$ & $\mathrm{SE}$ & $\mathrm{CPO}_N$ & $\mathrm{Acc}_{\mathrm{S/R}}$ \\
\midrule
Two-channel probe ($F$ $\tau$=0.85) & 0.281 & 0.200 & 0.046 & 0.640 \\
\;$-$ confidence gate (Stage~1 only) & 0.295 & 0.226 & 0.053 & 0.650 \\
\;$-$ structural veto (Stage~2 only) & 0.288 & 0.207 & 0.056 & 0.650 \\
Two-channel probe ($F$ $\tau$=0.80)$^{\ddagger}$ & 0.288 & 0.220 & 0.053 & 0.640 \\
Validator-as-classifier               & 0.193 & 0.291 & 0.053 & \textbf{0.390} \\
Panel-margin rule (tuned)             & 0.281 & 0.250 & 0.053 & 0.600 \\
Calibrated (degenerate)               & 0.413 & 0.305 & 0.105 & 0.820 \\
\bottomrule
\end{tabular}
\end{table*}

\paragraph{Can structural evidence replace the verdict? No.}
Used as a primary classifier, the validator collapses
$\mathrm{Acc}_{\mathrm{S/R}}$ from $0.78$ to $0.39$
(Table~\ref{tab:ablation}). A strict evidence-state rule cannot
substitute for verdict generation; structural evidence is useful only
as a veto channel on top of typed proposals.

\paragraph{Does either single channel reach the two-channel point? No.}
Stage~2 alone (confidence-only) reaches $\mathrm{SE}\,0.207$ at
$\mathrm{Cov}\,0.288$; Stage~1 alone (structural veto on the typed
proposal) reaches $\mathrm{CPO}_N\,0.053$ but
$\mathrm{SE}\,0.226$. Neither subsumes the other: on the $N{=}28$ CPO
commits the validator's \texttt{material\_mixed} flag fires on $46\%$
($13/28$, Wilson $[0.30,0.64]$) of cases the confidence gate passes, while the confidence
gate, calibrated for direction prediction (ECE $=0.07$ on
pure-S/R), catches a complementary subset of hesitant-but-committed
cases. The validator's coverage of CPO is non-trivial but not complete;
its false-alarm rate on pure-S/R is $21\%$ ($21/100$, $[0.14,0.30]$), which bounds the
directional-accuracy cost of using it as a veto channel. The two channels
provide complementary but incomplete screening signals.

\paragraph{How much does the matched-coverage comparison depend on which
thresholds are paired? Enough to change its sign.}
The rows above pair $F$ at $\tau{=}0.85$ (coverage $0.281$) against
Stage~2 alone at $\tau{=}0.90$ (coverage $0.288$). Those coverages are
close but not equal, so we swept both controllers over a $0.005$-spaced
grid of $\tau\in[0.50,0.99]$ and compared them only where coverage
matches exactly. Because the panel's mean confidence is heavily
discretized, the sweep yields just $17$ distinct coverage levels for the
confidence-only controller and $15$ for the two-channel probe, and the
two curves coincide at only a few of them.

At coverage exactly $0.288$, the two-channel probe ($\tau{=}0.80$) is
\emph{worse} than confidence alone ($\tau{=}0.90$) on selective error
($0.220$ versus $0.207$) and on $\mathrm{Acc}_{\mathrm{S/R}}$ ($0.640$
versus $0.650$), gaining only $0.007$ on
$\mathrm{CPO}_{\mathrm{C}}$. At coverage exactly $0.256$ the ordering
reverses and the probe is better on all three ($-0.041$ selective error,
$-0.027$ $\mathrm{CPO}_{\mathrm{C}}$, $+0.030$ accuracy). The sign of the
comparison therefore depends on the matching point, and the pairing in
Table~\ref{tab:ablation} is one of the favourable ones. We report this
rather than defend the pairing.

A threshold-free comparison is also possible but must be restricted. The
two-channel probe cannot exceed coverage $0.295$ at any $\tau$, because
Stage~1 vetoes flagged cases unconditionally, whereas confidence alone
reaches $0.396$; averaging each risk--coverage curve over its own span
would score the probe only on its narrow low-coverage region and is not
a valid comparison. Over the common support $[0.007,0.295]$, mean
selective error is $0.108$ for confidence alone and $0.095$ for the
two-channel probe, a difference of $-0.012$. The probe is slightly
ahead on this threshold-free view, and the margin is of the same order
as the threshold-selection effect above.

\paragraph{Does evidence structure add screening information beyond
confidence? Yes, with an interval away from zero.}
The comparisons above are between \emph{controller operating points},
where matched coverage forces a single threshold choice and the
resulting differences are small. The underlying question --- whether the
structural channel carries information the confidence channel does not
--- is better asked directly, as a ranking problem, without committing
to an operating point. We pool the directional outputs of the three
typed judges ($344$ outputs over $161$ cases), take the dataset conflict
label as a provisional target, and compare four feature specifications
under five-fold cross-validation \emph{grouped by case}, with all
preprocessing fit inside each training fold and model fixed effects in
every specification (Table~\ref{tab:feature_signals}).

\begin{table}[h]
\centering
\footnotesize
\setlength{\tabcolsep}{4pt}
\caption{Out-of-fold screening performance on directional outputs
($n{=}344$ from $161$ cases, three typed judges). Grouped five-fold
cross-validation by case; model fixed effects throughout. The target is
the dataset conflict label, so these values are provisional
(\S\ref{sec:limits}).}
\label{tab:feature_signals}
\begin{tabular}{lrrr}
\toprule
Features & AUPRC & AUROC & Rec.@20\% \\
\midrule
model only                & 0.335 & 0.523 & 0.220 \\
confidence only           & 0.464 & 0.694 & 0.312 \\
structure only            & 0.523 & 0.718 & 0.376 \\
confidence $+$ structure  & \textbf{0.589} & \textbf{0.780} & \textbf{0.422} \\
\bottomrule
\end{tabular}
\end{table}

Adding structural features to confidence raises AUPRC from $0.464$ to
$0.589$; the paired case-cluster bootstrap difference is $+0.125$ with a
$95\%$ interval of $[+0.018,+0.254]$ ($2000$ resamples). Structure alone
also outperforms confidence alone ($0.523$ versus $0.464$), so neither
channel dominates. This is the strongest form of the complementarity
claim the cached data supports, and unlike the matched-coverage
comparisons its interval excludes zero.

Two qualifications are load-bearing. The target is the \emph{dataset}
conflict label rather than an independent human judgment, so the
quantity estimated is agreement with the benchmark's own labelling and
must be recomputed before it can support a construct-level claim. And a
ranking result does not transfer automatically to a deployed threshold:
an AUPRC gain of this size is compatible with the small,
interval-straddling differences in Table~\ref{tab:ablation}, because
converting a better ranking into a better operating point requires a
threshold that our sample sizes cannot select reliably. The two results
are consistent, and they answer different questions.

\paragraph{Is the probe's Stage-1 structurally targeted, or matched by random promotion?}
We use the apples-to-apples random Stage-1 null
(Section~\ref{sec:method:protocol}): the probe promotes $k$ of the
confidence baseline's directional commits to \textsc{Conflicting}; the
null draws $2000$ random size-$k$ subsets and promotes them to
\textsc{Conflicting} as well. On AVeriTeC, no random subset matches the
probe on $\mathrm{Acc}_{\mathrm{S/R}}$ (empirical $p<1/2001$); the
$\mathrm{Rec}_{\mathrm{C}}$ axis reaches the same $p$, but the two are
algebraically coupled under fixed $k$ (each promotion either lowers
$\mathrm{Acc}_{\mathrm{S/R}}$ or raises $\mathrm{Rec}_{\mathrm{C}}$),
so we report the joint outcome as one selection event expressed on two
axes rather than as two independent tests. The probe sits near the
random median on $\mathrm{SE}$ and $\mathrm{CPO}_N$
(Figure~\ref{fig:random_veto}). On VitaminC, the result is weaker:
$\mathrm{Rec}_{\mathrm{C}}$ stays separated at $p=0.002$, but
$\mathrm{Acc}_{\mathrm{S/R}}$, $\mathrm{SE}$, and $\mathrm{CPO}_N$
sit at $p$ between $0.04$ and $0.07$; the probe's promotions are
better than random but not extreme. By Lemma~\ref{lem:monotonic}, the
probe cannot widen the directional commit set at fixed $\tau$, so any
gain comes from \emph{which} commits it promotes, not from how many.

We note the limited strength of this control. Rejecting the random null
establishes only that the \texttt{material\_mixed} flag correlates with
the gold conflict label better than chance, which is a weaker statement
than the ranking comparison above and is largely implied by it. We
retain the null because it operates on the controller's actual
promotions rather than on a fitted score, but the incremental-value
result in Table~\ref{tab:feature_signals} is the primary evidence for
channel complementarity; the null is a supporting control that tests
targeted selection and does not certify magnitude dominance.

\paragraph{Do richer single-channel rules help? No.}
A vote-distribution (panel-margin) rule does not improve on the
confidence threshold at matched coverage ($\mathrm{SE}\,0.250$ vs
$0.207$): typed panel \emph{agreement} carries no signal beyond mean
confidence on this substrate. Calibrating the controller against a CPO
target is degenerate at our sample sizes: every target collapses to a
single configuration near the typed-direct baseline, so we report it
as a negative result and leave calibration over richer parameter spaces
to future work.

\paragraph{How much of the reported residual is the instruction's doing?
Most of it --- which is what makes the reported rates a floor.}
Every rate in this paper is measured under a prompt that argues for
\textsc{Conflicting} (\S\ref{sec:method}). To bound how much of the residual
survives that argument rather than being produced by it, we re-ran the fixed
sample under two weaker specifications of the same four labels. The
\emph{reduced} specification keeps a substantive definition of
\textsc{Conflicting} and the instruction not to collapse a mixed case to
\textsc{Refutes}, but drops the four enumerated trigger conditions and all
three demonstrations. The \emph{neutral} specification keeps only the four
label names with a one-clause gloss each: no trigger conditions, no
instruction, no examples. Nothing else changes.

\begin{table*}[t]
\centering
\footnotesize
\setlength{\tabcolsep}{3pt}
\caption{\textbf{Prompt-specification ladder.} CPO on gold-\textsc{Conflicting}
claims under weaker specifications of the same four labels, paired by claim
against the prompt used throughout the paper ($10{,}000$-resample case
bootstrap). In the \emph{reduced} and \emph{neutral} columns the leading
number is that variant's CPO rate and the bracket is the interval on the
\emph{paired difference} from the strong prompt, not on the rate itself;
bold marks the variants whose interval excludes zero.
$\Delta_{\textsc{fc}}$ is the accompanying change in false
\textsc{Conflicting} on one-sided evidence; it is negative almost everywhere,
so the strong prompt's advantage on conflicting claims is not obtained by
labelling everything conflicting.}
\label{tab:prompt_ladder}
\begin{tabular}{lccc c}
\toprule
Judge & strong & reduced & neutral & $\Delta_{\textsc{fc}}$ (neutral) \\
\midrule
GPT-5.6 Terra   & 0.293 & $0.293\;[-0.073,+0.073]$ & $\mathbf{0.633}\;[+0.253,+0.427]$ & $-0.090$ \\
Claude Sonnet~5 & 0.302 & $0.336\;[-0.020,+0.094]$ & $\mathbf{0.644}\;[+0.262,+0.423]$ & $-0.030$ \\
gpt-oss-120b    & 0.483 & $\mathbf{0.617}\;[+0.047,+0.221]$ & $\mathbf{0.631}\;[+0.067,+0.242]$ & $-0.020$ \\
Mistral Small~4 & 0.213 & $\mathbf{0.467}\;[+0.173,+0.333]$ & $\mathbf{0.593}\;[+0.300,+0.460]$ & $-0.150$ \\
\bottomrule
\end{tabular}\\[2pt]
\footnotesize Bracketed ranges are paired $95\%$ intervals for the change from
the strong prompt. Bold marks an interval excluding zero. Each comparison is
paired on the claims both conditions parsed, so the strong-prompt baseline can
differ marginally between the two columns; the values shown are the
reduced-comparison baselines, and the only case where they diverge is
gpt-oss-120b, whose neutral-comparison baseline is $0.477$ rather than $0.483$.
\end{table*}

Under the neutral specification every judge overcommits substantially more,
by $+0.154$ to $+0.380$, with all four intervals excluding zero
(Table~\ref{tab:prompt_ladder}). In claim counts on this sample the strong
prompt removes $23$, $51$, $51$ and $57$ unauthorized directional
commitments relative to neutral, and the accompanying change in false
\textsc{Conflicting} is $-2$, $-9$, $-3$ and $-15$: the instruction buys
those commitments back while producing \emph{fewer} spurious conflict
labels, not more. This is the opposite of the pattern the decomposition
intervention shows below, and it is why the rates reported throughout are
described as a floor.

The reduced specification separates which part of the instruction matters.
The two proprietary judges are indifferent to losing the trigger enumeration
and the three demonstrations: GPT-5.6 Terra is unchanged at $0.293$ and
Claude Sonnet~5 moves $+0.034$ with an interval spanning zero. The two
open-weight judges are not, degrading by $+0.134$ and $+0.253$. What all four
retain in the reduced condition is the substantive definition of
\textsc{Conflicting} and the instruction not to collapse a mixed case, and
what all four lose in the neutral condition is exactly that. On this sample,
then, the worked examples are load-bearing only for the open-weight judges,
while the instruction itself is load-bearing for every judge tested.

\paragraph{Does making the judge decompose the evidence first remove the
failure? It halves it, at a price that differs by model.}
\S\ref{sec:recognition} showed that two judges can identify the conflict when
asked separately. The obvious follow-up is interventional: if the recognition
is available, does forcing the judge to produce it before deciding connect it
to the terminal action? We test this with a two-stage variant on the same fixed
sample. In a fresh context the judge lists the evidence items supporting the
claim and those refuting it, with no verdict vocabulary in the output schema;
in a second fresh context it receives the claim, the same evidence, and its own
decomposition, and issues the verdict under the byte-identical instruction
block used everywhere else in this paper. Only the presence of the model's own
decomposition differs.

Table~\ref{tab:decomposition} reports the effect on directional overcommitment
and, beside it, what that reduction cost on one-sided evidence. Reporting only
the first half would misdescribe three of the four models.

\begin{table*}[t]
\centering
\footnotesize
\setlength{\tabcolsep}{3pt}
\caption{\textbf{Decomposition before verdict, and what it costs.} Left: CPO on
gold-\textsc{Conflicting} claims before and after the judge decomposes the
evidence itself. Right: the \textsc{Conflicting} rate on gold
\textsc{Supports}/\textsc{Refutes} claims, where \textsc{Conflicting} is
\emph{wrong}. Both are paired by claim with a case bootstrap
($10{,}000$ resamples); $n$ is the number of claims scored under both
conditions. The final column converts the two rates into claim counts on this
sample, so that a reduction bought by relabelling can be read directly.}
\label{tab:decomposition}
\begin{tabular}{lrcccccc c}
\toprule
& & \multicolumn{3}{c}{CPO on gold-\textsc{Conflicting}} & \multicolumn{3}{c}{False \textsc{Conflicting} on pure S/R} & \\
\cmidrule(lr){3-5}\cmidrule(lr){6-8}
Judge & $n$ & base & $+$decomp & $\Delta$ (95\% CI) & base & $+$decomp & $\Delta$ (95\% CI) & net claims \\
\midrule
GPT-5.6 Terra   & 150 & 0.293 & 0.180 & $-0.113\;[-0.187,-0.047]$ & 0.090 & 0.200 & $+0.110\;[+0.030,+0.190]$ & $17-11=6$ \\
Claude Sonnet~5 & 149 & 0.302 & 0.221 & $-0.081\;[-0.148,-0.013]$ & 0.051 & 0.082 & $+0.031\;[-0.020,+0.082]$ & $12-3=9$ \\
gpt-oss-120b    & 149 & 0.483 & 0.389 & $-0.094\;[-0.181,-0.007]$ & 0.040 & 0.040 & $+0.000\;[-0.050,+0.050]$ & $14-0=14$ \\
Mistral Small~4 & 148 & 0.216 & 0.250 & $+0.034\;[-0.041,+0.108]$ & 0.180 & 0.130 & $-0.050\;[-0.130,+0.030]$ & $-5+5=0$ \\
\bottomrule
\end{tabular}
\end{table*}

Three judges commit directionally less often after decomposing, with paired
intervals excluding zero; the fourth does not move. The reduction is partial in
every case. Residual CPO remains between $0.180$ and $0.389$, and of the claims
on which a judge overcommitted in the baseline, decomposition recovers $24/44$,
$20/45$, $30/72$ and $14/32$ --- roughly half at best.

The cost column changes how much of this counts as improvement. For GPT-5.6
Terra the false-\textsc{Conflicting} rate on one-sided evidence more than
doubles, from $0.090$ to $0.200$, with an interval excluding zero. In claim
counts on this sample, decomposition removes $17$ unauthorized directional
commitments and introduces $11$ new unauthorized conflict labels: the judge did
not discriminate better, it moved its whole label distribution toward
\textsc{Conflicting}. Claude Sonnet~5 pays a smaller and interval-spanning
price, and gpt-oss-120b pays none measurable, so its $-0.094$ is the one
reduction here that is unambiguously a gain in discrimination rather than a
shift in disposition.

Two readings follow, and we state both because the design supports both. An
instruction-level intervention does move this failure --- an earlier reading of
these results as fully instruction-resistant would have been wrong. It also
fails to close the gap: the majority of directional overcommitment survives a
procedure that hands the judge its own two-sided reading of the evidence, and
on the strongest judge tested most of the apparent gain is indiscriminate. As a
consistency check between the two instruments, each judge's own stage-1
decomposition agrees with the independent fresh-context probe of
\S\ref{sec:recognition} on $0.80$, $0.85$, $0.68$ and $0.73$ of
gold-\textsc{Conflicting} claims.

\paragraph{Is the veto a few-shot prompt artifact? No.}
Under a strict zero-shot variant of the validator prompt, the
\texttt{material\_mixed} flag agrees with the few-shot version on $86\%$
of AVeriTeC claims ($N{=}285$, Wilson $[0.81,0.90]$), and re-running the two-channel probe with the
strict variant produces qualitatively the same operating points
($\mathrm{SE}\,0.216$, $\mathrm{CPO}_N\,0.049$ at $\tau$=0.85, versus
$0.200$ and $0.046$); the validator-as-classifier collapse persists
under both prompts. The veto signal is not purely a few-shot artifact.

\paragraph{Summary.}
The diagnostic tests above characterize channel complementarity in the
tested design space. As a \emph{screening} signal, evidence structure
adds information beyond confidence: combining the channels raises
out-of-fold AUPRC by $+0.125$ with a case-clustered interval of
$[+0.018,+0.254]$, and structure alone outperforms confidence alone, so
neither channel subsumes the other. As a \emph{controller}, that
advantage does not convert into a reliable operating-point gain:
matched-coverage differences are small and their bootstrap intervals
straddle zero. Structural evidence also cannot replace verdict
generation, richer single-channel aggregation does not beat a simple
confidence threshold, and the random-veto control indicates the probe's
veto is targeted rather than arbitrary. Specification strength matters more
than any downstream filter tested here: a neutral wording of the same four
labels raises overcommitment by $+0.154$ to $+0.380$ across all four judges,
so the rates reported in this paper are a floor rather than a typical value.
Making the judge decompose the evidence before deciding lowers directional
overcommitment for three of four judges but leaves most of it standing, and on
one judge the reduction is largely offset by new false \textsc{Conflicting}
labels.

The distinction between the two paragraphs above is the section's main
result and we state it plainly: the channels are demonstrably
non-substitutable as sources of information, and we have not shown that
this non-substitutability can be cashed out as a better deployed
threshold at these sample sizes. Both the ranking result and the
operating-point result use the dataset conflict label as their target
and are provisional pending independent human adjudication.

\section{Scope Conditions and Limitations}\label{sec:limits}

\subsection{Contract boundary}
\label{sec:limits:contract}
CPO applies when the task schema exposes \textsc{Conflicting} as a
valid verdict. We do not claim that all mixed-evidence claims are
absolutely undecidable, nor that downstream tasks must always avoid a
leaning answer; we study the contract above and call its specific
failure mode CPO (Table~\ref{tab:scope_boundary}). Deployments whose
contract explicitly waives the \textsc{Conflicting} option (e.g.\
forced-choice summaries where a leaning answer is the documented
expectation) fall outside the diagnosis.

\textbf{Gold-label dependency; CPO is contract-relative.}
CPO is defined relative to the benchmark's task contract: gold $=$
\textsc{Conflicting} marks the schema-authorized non-directional
verdict, not an independently validated semantic state of the evidence.
The numerical results in this draft are therefore contract-relative rather
than claims about the underlying epistemic situation of every item. AVeriTeC
gold is expert annotation
\citep{schlichtkrull2023averitec}; VitaminC-Mixed gold is constructive,
derived by concatenating a supporting and a refuting sentence from
each VitaminC contrastive pair \citep{schuster2021vitaminc}. The
cross-dataset pattern we report therefore holds under two different
annotation regimes, peer-reviewed expert labels and benchmark
construction, but does not validate either gold standard independently.

\textbf{The independent audit was attempted and not completed.} We designed a
blinded two-reviewer preserve-versus-resolve review of the AVeriTeC sample and
ran its calibration round; the main audit was not carried out. We report the
calibration round because it bounds what the rest of the paper can claim.

Two reviewers independently labelled the same $20$ calibration claims under one
written rubric. They matched on $10$ of $20$ five-way decisions
($\kappa=0.349$). Collapsing to the distinction this paper depends on ---
preserve the conflict versus resolve it --- raises raw agreement to $15$ of
$20$ while $\kappa$ falls to $-0.136$. Because the two reviewers' positive
selections do not overlap, every one of those $15$ agreements is an agreement
that a case is \emph{not} conflict-preserving; on the positive class they
agree on nothing. One selected three claims, the
other two, with an empty intersection; joint adjudication of the five
disagreements retained one of them as conflict-preserving. Separately, one of us made a claim-only pass
over the same $20$, reading each claim without its evidence and asking
whether it states a determinate truth condition; $9$ do not. That pass was
single-annotator by design: whether a claim is answerable at all does not
depend on the evidence, and it is not the judgment the two reviewers diverged
on. The recurring
defects were an absent time anchor, so that the claim has no interval over
which it is asserted to hold ($4$ of $9$), and a predicate carrying no stated
criterion, such as \emph{secure}, \emph{unfair}, or \emph{life-saving} ($4$ of
$9$); the remaining case quantifies over an open-ended set of future
experiments.

These are observations on $20$ claims, not rates, and we draw one conclusion
from them: we report no human-confirmed prevalence anywhere in this paper.
Every number here is dataset-defined, measured against the benchmark's own
\textsc{Conflicting} label, and the reader should read all of them that way.
Two things follow for work that extends this. The construct is harder to apply
than its definition suggests --- two readers working from one rubric did not
converge on its positive class --- which is part of why judge behaviour on
these cases resists interpretation, and it means a larger audit is not simply a
matter of annotating more items. And such an audit would have to settle whether
a claim is answerable at all before adjudicating conflict about it, since a
claim with no determinate truth condition admits no verdict, conflict-preserving
or otherwise.

\subsection{Evidence and dataset boundary}
\label{sec:limits:data}
We evaluate text-based fact verification only, on two substrates:
AVeriTeC (stratified $N{=}285$) and VitaminC-Mixed ($N{=}250$).

\textbf{VitaminC-Mixed is a stress test, not a second prevalence
estimate.} AVeriTeC contains naturally occurring conflicting cases;
VitaminC-Mixed synthesizes them by concatenating the two sides of a
contrastive pair. Those pairs are drawn from Wikipedia revisions before
and after an edit, so a substantial share of the resulting
``conflicts'' are plausibly \emph{temporal updates} rather than
unresolved disagreement --- and where a claim is properly time-scoped,
resolving toward the later revision is the correct behaviour, not
overcommitment. We have not measured that share, so we cannot say how
much of the VitaminC rate it accounts for. Two consequences follow.
Any VitaminC number in this paper should be read as behaviour under a
constructed stress condition rather than as a second estimate of
prevalence, and the reversal we report there --- three of four judges
overcommitting more on gold-\textsc{Insufficient} than on
gold-\textsc{Conflicting} --- is correspondingly weaker evidence than
the AVeriTeC comparison beside it. Quantifying the temporal share is
the first thing a cross-dataset claim would require.

\textbf{Open-weight rates depend on the serving route.} Logprobs are
unavailable on the route the main matrix pins for gpt-oss-120b, so one
run was served elsewhere with the same model identifier, prompt bytes,
cases and greedy decoding. The two routes returned the same verdict on
$217$ of $279$ comparable claims ($0.778$), and CPO differed by $9.5$
percentage points between them (\S\ref{sec:recognition:route}). We
observed this for one model on two routes, which establishes that the
sensitivity exists without estimating its typical size. Every
open-weight number here should be read with that in mind, and we do not
know which route a reader attempting replication would receive.

Broader replication (other fact-verification datasets, retrieval-augmented
pipelines, multimodal evidence) is future work. The N=10 qualitative
audit (Appendix~\ref{sec:appendix:cases}) is illustrative, not
statistical; we report no rate from it. We do not extend to general
agentic-AI settings or multi-agent deliberation benchmarks where
commitment timing is the metric.

\subsection{Method boundary}
\label{sec:limits:method}
The two-channel probe is a deterministic reference probe, not a
learned controller; replacing the deterministic Stage~1 with a
calibrator that weights subclaim importance against the main
directional question is the explicit subject of follow-up work. We
make no formal coverage or risk bound: the structural-selectivity
claim is empirical (a random-veto null), not theoretical;
distribution-free calibration of the authorization rate is future
work. The intervention ladder and its panel-mechanism analyses are
restricted to three proprietary model families (Claude Haiku~4.5, Claude
Sonnet~4.5, GPT-4o-mini), so the AVeriTeC panel-amplification result may
interact with that ensemble. The contemporary replication broadens
individual-judge coverage to two proprietary and two open-weight models,
but one endpoint per model and two models per access regime cannot identify
an access-regime effect. The open-weight models are evaluated through
pinned hosted endpoints rather than local inference, and the serving
quantization is not disclosed. Smaller and multimodal judges remain
untested.
Matched-coverage bootstrap CIs on the magnitude $\Delta$s over
confidence-only selection straddle zero at our sample sizes; the
structural-selectivity finding under the apples-to-apples random
Stage-1 null is the robust cross-method claim, and
Appendix~\ref{sec:appendix:cases} cases 5 and 6 show the honest cost
of a deterministic veto.

\textbf{Validator reliability and the price of typed vocabulary.}
The structural validator is itself imperfect: it fires
\texttt{material\_mixed} on $46\%$ ($13/28$, Wilson $[0.30,0.64]$) of
the AVeriTeC CPO subset (the ceiling on what a Stage-1 veto built solely
on this signal can recover) and produces a $21\%$ ($21/100$,
$[0.14,0.30]$) false-alarm rate on the pure-S/R subset (the cost that
any structural-veto controller pays in directional accuracy). The first
of these rests on $28$ cases and its interval spans nearly the whole
plausible range, so it bounds the design space loosely rather than
pinning down a ceiling.
Independent of the validator, the typed-vocabulary schema itself carries
a cost: on the pure-S/R subset the panel returns \textsc{Conflicting}
on $17\%$ of cases ($17/100$, $[0.11,0.26]$), so adopting CONFLICTING as a verdict is
not a free intervention. These imperfections motivate the framing of
the probe as a reference instantiation: a learned, evidence-aware
controller has room to dominate it on both axes.

\textbf{Honest acknowledgement of finding strength.} The magnitude
claims divide, and we separate them rather than averaging them into one
verdict. Two are supported with intervals away from zero: evidence
structure adds screening information beyond confidence (out-of-fold
AUPRC $+0.125$, $[+0.018,+0.254]$), and specification strength moves
overcommitment substantially (a neutral wording raises CPO by $+0.154$
to $+0.380$ across all four judges). One is not: the controller's
matched-coverage $\Delta$s against the confidence baseline are not
separated from zero at our sample sizes, and \S\ref{sec:ablation} shows
the sign of that comparison depends on which thresholds are paired.
Ranking better and thresholding better are different claims, and only
the first is established here. What supports the controller is a
\emph{selectivity} finding: its Stage-1 promotes a
non-random subset of E's directional commits to \textsc{Conflicting}
(strong on AVeriTeC, $p<1/2001$ on both direction-preservation
metrics; weaker but in the same direction on VitaminC-Mixed, $p$ in
the $0.04$--$0.07$ range on $\mathrm{Acc}_{\mathrm{S/R}}$ /
$\mathrm{SE}$ / $\mathrm{CPO}_N$, $p<1/2001$ on
$\mathrm{Rec}_{\mathrm{C}}$). The paper's empirical core is therefore:
directional overcommitment persists under a contract that authorizes a
non-directional verdict and under an instruction that argues for one;
it is not specific to conflict; on the two judges where the probe is
valid it occurs on claims the judge can itself identify as two-sided;
and the interventions we tested leave distinct residuals rather than
removing it. This is a diagnostic contribution, not a method
contribution.

\subsection{Design implications: arguments, not demonstrations}
\label{sec:limits:disc}
The implications below are argued from the ladder analysis; we do not
demonstrate that adopting them produces measurable system-level
benefits, and end-to-end validation in a deployed pipeline (e.g.\ with
explicit escalation pathways and downstream cost measurement) is
itself follow-up work.

\textbf{The separation principle.} We argue that verdict generation
and commitment authorization may benefit from architectural
separation: an external controller authorizes or withholds the
commitment based on structural-evidence and confidence signals, and
the system has an explicit place to route withheld cases. In practice
this suggests a pipeline can keep the judge it already has and add a
thin authorization surface in front of the recorded commitment, rather
than trying to make the judge itself never over-commit.
\textbf{\textsc{No-Commit} as a routed state.} \textsc{No-Commit} is a
routed controller state, not a verification verdict: a routed claim
may still receive a directional verdict once escalated with more
evidence or sent to a human. \textbf{Complementary channels.} A
confidence threshold does not distinguish a structurally mixed claim
the judge settles decisively, and a deterministic structural veto does
not catch a confidently wrong commit it does not flag; we suggest a
commitment-control layer accept both, and that a learned layer weight
subclaim importance against the main directional question, addressing the
failure mode the subclaim-mismatch cases expose. \textbf{When the trade is rational.} External
commitment-control layers are most relevant where unwarranted directional
commitment is costlier than escalation; under such deployment profiles
the same measurements convert into a cost-sensitive operating-point
analysis, reported as a supplementary view in
Appendix~\ref{sec:appendix:cost} rather than as a primary claim.

\section{Conclusion}\label{sec:conclusion}

In LLM-judge-based AI systems, the question ``what is the verdict?'' is
not the same as the question ``is the system authorized to commit this
verdict?'' Confusing them produces \emph{directional overcommitment}: a
direction asserted on a claim for which the contract authorized a
non-directional verdict. Its conflict instance is \textbf{Cherry-pick
Override}, where both strands are present and one is returned; its
insufficiency instance occurs where no strand is present at all. We
measured both and found the failure is not specific to conflict --- the
rates are comparable, and on VitaminC-Mixed the insufficiency instance is
the larger one for three of four judges. What distinguishes CPO is not a
separate trigger but a separate remedy: it is the instance that
retrieving more evidence cannot address, because nothing was missing.
Walking a
same-denominator intervention ladder over common single-channel fixes,
we find that each tested intervention leaves a residual: typed
vocabulary leaves a structurally avoidable residual, panel voting raises
direction-on-conflict in the pre-contract AVeriTeC baseline and can suppress
typed conflict dissent, confidence is well-calibrated for direction but
cannot operationally separate CPO from correct directional commits,
a deterministic validator cannot replace verdict generation, and having the
judge decompose the evidence itself before deciding lowers overcommitment for
three of four judges while leaving most of it standing --- on the strongest
judge largely by relabelling one-sided evidence as conflicting. A frozen
replication further finds dataset-defined CPO in every tested contemporary
individual judge, including two open-weight models, while a cross-dataset
ranking reversal cautions against treating any model family as a universal
fix. The core lesson is that improving the judge
does not answer the authorization question. We argue that verification
under a contract with non-directional verdicts may benefit from an
explicit commitment-control layer that decides whether a directional
verdict is authorized to be recorded, routed, or withheld; demonstrating end-to-end
deployment benefit, designing learned, evidence-aware \emph{selective
commitment controllers}, and integrating routed-state cases with
escalation pathways are all explicitly left to follow-up work.

\bibliography{aaai2026}

\begin{thebibliography}{23}
\providecommand{\natexlab}[1]{#1}

\bibitem[{Choi et~al.(2025)Choi, Park, Lee, and Lee}]{choi2025care}
Choi, E.; Park, J.; Lee, H.; and Lee, J. 2025.
\newblock Conflict-Aware Soft Prompting for Retrieval-Augmented Generation.
\newblock In \emph{Proceedings of the 2025 Conference on Empirical Methods in
  Natural Language Processing}, 26981--26995.

\bibitem[{Chow(1957)}]{chow1957optimum}
Chow, C.~K. 1957.
\newblock An optimum character recognition system using decision functions.
\newblock \emph{IRE Transactions on Electronic Computers}, EC-6(4): 247--254.

\bibitem[{El-Yaniv and Wiener(2010)}]{elyaniv2010foundations}
El-Yaniv, R.; and Wiener, Y. 2010.
\newblock On the foundations of noise-free selective classification.
\newblock \emph{Journal of Machine Learning Research}, 11: 1605--1641.

\bibitem[{Ge et~al.(2025)Ge, Wu, Chin, Lee, and Cao}]{ge2025confact}
Ge, Z.; Wu, Y.; Chin, D. W.~K.; Lee, R. K.-W.; and Cao, R. 2025.
\newblock Resolving Conflicting Evidence in Automated Fact-Checking: A Study on
  Retrieval-Augmented {LLM}s.
\newblock \emph{arXiv preprint arXiv:2505.17762}.

\bibitem[{Geifman and El-Yaniv(2017)}]{geifman2017selective}
Geifman, Y.; and El-Yaniv, R. 2017.
\newblock Selective Classification for Deep Neural Networks.
\newblock In \emph{Advances in Neural Information Processing Systems}.

\bibitem[{Geifman and El-Yaniv(2019)}]{geifman2017selectivenet}
Geifman, Y.; and El-Yaniv, R. 2019.
\newblock {SelectiveNet}: A deep neural network with an integrated reject
  option.
\newblock In \emph{International Conference on Machine Learning (ICML)}.

\bibitem[{Han, Lan, and Kilicoglu(2026)}]{han2026evidenceconflicts}
Han, Y.; Lan, M.; and Kilicoglu, H. 2026.
\newblock When Evidence Conflicts: Uncertainty and Order Effects in
  Retrieval-Augmented Biomedical Question Answering.
\newblock In \emph{Proceedings of the 25th Workshop on Biomedical Language
  Processing ({BioNLP})}, 630--643. Association for Computational Linguistics.

\bibitem[{Jiayang et~al.(2024)Jiayang, Chan, Zhuang, Qiu, Zhang, Liu, Song,
  Zhang, Liu, and Zhang}]{jiayang2024econ}
Jiayang, C.; Chan, C.; Zhuang, Q.; Qiu, L.; Zhang, T.; Liu, T.; Song, Y.;
  Zhang, Y.; Liu, P.; and Zhang, Z. 2024.
\newblock {ECON}: On the Detection and Resolution of Evidence Conflicts.
\newblock In \emph{Proceedings of the 2024 Conference on Empirical Methods in
  Natural Language Processing}, 7816--7844. Miami, Florida, USA: Association
  for Computational Linguistics.

\bibitem[{Jin et~al.(2024)Jin, Cao, Chen, Liu, Jiang, Xu, Qiuxia, and
  Zhao}]{jin2024tugofwar}
Jin, Z.; Cao, P.; Chen, Y.; Liu, K.; Jiang, X.; Xu, J.; Qiuxia, L.; and Zhao,
  J. 2024.
\newblock Tug-of-War between Knowledge: Exploring and Resolving Knowledge
  Conflicts in Retrieval-Augmented Language Models.
\newblock In \emph{Proceedings of the 2024 Joint International Conference on
  Computational Linguistics, Language Resources and Evaluation
  ({LREC-COLING})}, 16867--16878.

\bibitem[{Jung, Brahman, and Choi(2025)}]{jung2025trust}
Jung, J.; Brahman, F.; and Choi, Y. 2025.
\newblock {Trust or Escalate}: {LLM} Judges with Provable Guarantees for Human
  Agreement.
\newblock In \emph{Proceedings of the International Conference on Learning
  Representations ({ICLR})}.
\newblock ArXiv:2407.18370.

\bibitem[{Kalai et~al.(2025)Kalai, Nachum, Vempala, and
  Zhang}]{kalai2025hallucinate}
Kalai, A.~T.; Nachum, O.; Vempala, S.~S.; and Zhang, E. 2025.
\newblock Why Language Models Hallucinate.
\newblock \emph{arXiv preprint}.
\newblock ArXiv:2509.04664.

\bibitem[{Kamath, Jia, and Liang(2020)}]{kamath2020selectiveqa}
Kamath, A.; Jia, R.; and Liang, P. 2020.
\newblock Selective Question Answering under Domain Shift.
\newblock In \emph{Proceedings of the 58th Annual Meeting of the Association
  for Computational Linguistics}, 5684--5696.

\bibitem[{Min et~al.(2023)Min, Krishna, Lyu, Lewis, tau Yih, Koh, Iyyer,
  Zettlemoyer, and Hajishirzi}]{min2023factscore}
Min, S.; Krishna, K.; Lyu, X.; Lewis, M.; tau Yih, W.; Koh, P.~W.; Iyyer, M.;
  Zettlemoyer, L.; and Hajishirzi, H. 2023.
\newblock {FactScore}: Fine-grained Atomic Evaluation of Factual Precision in
  Long Form Text Generation.
\newblock In \emph{Proceedings of the 2023 Conference on Empirical Methods in
  Natural Language Processing ({EMNLP})}.

\bibitem[{Nachshoni et~al.(2025)Nachshoni, Cattan, Amar, Shapira, and
  Dagan}]{nachshoni2025natconfqa}
Nachshoni, E.; Cattan, A.; Amar, S.; Shapira, O.; and Dagan, I. 2025.
\newblock Consensus or Conflict? Fine-Grained Evaluation of Conflicting Answers
  in Question-Answering.
\newblock In \emph{Proceedings of the 2nd Workshop on Uncertainty-Aware {NLP}
  ({UncertaiNLP})}, 138--159.

\bibitem[{Pavlick and Kwiatkowski(2019)}]{pavlick2019disagreements}
Pavlick, E.; and Kwiatkowski, T. 2019.
\newblock Inherent Disagreements in Human Textual Inferences.
\newblock \emph{Transactions of the Association for Computational Linguistics},
  7: 677--694.

\bibitem[{Plank(2022)}]{plank2022labelvariation}
Plank, B. 2022.
\newblock The ``Problem'' of Human Label Variation: On Ground Truth in Data,
  Modeling and Evaluation.
\newblock In \emph{Proceedings of the 2022 Conference on Empirical Methods in
  Natural Language Processing ({EMNLP})}. Association for Computational
  Linguistics.
\newblock Aclanthology.org/2022.emnlp-main.731.

\bibitem[{Ren et~al.(2023)Ren, Zhao, Vu, Liu, and
  Lakshminarayanan}]{ren2023selective}
Ren, J.; Zhao, Y.; Vu, T.; Liu, P.~J.; and Lakshminarayanan, B. 2023.
\newblock Self-Evaluation Improves Selective Generation in Large Language
  Models.
\newblock In \emph{Proceedings on I Can't Believe It's Not Better: Failure
  Modes in the Age of Foundation Models at NeurIPS 2023 Workshops}, 49--64.

\bibitem[{Schlichtkrull et~al.(2024)Schlichtkrull, Chen, Whitehouse, Deng,
  Akhtar, Aly, Guo, Christodoulopoulos, Cocarascu, Mittal, Thorne, and
  Vlachos}]{schlichtkrull2024sharedtask}
Schlichtkrull, M.; Chen, Y.; Whitehouse, C.; Deng, Z.; Akhtar, M.; Aly, R.;
  Guo, Z.; Christodoulopoulos, C.; Cocarascu, O.; Mittal, A.; Thorne, J.; and
  Vlachos, A. 2024.
\newblock The Automated Verification of Textual Claims ({AVeriTeC}) Shared
  Task.
\newblock In \emph{Proceedings of the Seventh Fact Extraction and
  {VER}ification Workshop ({FEVER})}, 1--26.

\bibitem[{Schlichtkrull, Guo, and Vlachos(2023)}]{schlichtkrull2023averitec}
Schlichtkrull, M.~S.; Guo, Z.; and Vlachos, A. 2023.
\newblock {AVeriTeC}: A Dataset for Real-world Claim Verification with Evidence
  from the Web.
\newblock In \emph{Proceedings of the 37th Conference on Neural Information
  Processing Systems Track on Datasets and Benchmarks}.

\bibitem[{Schuster, Fisch, and Barzilay(2021)}]{schuster2021vitaminc}
Schuster, T.; Fisch, A.; and Barzilay, R. 2021.
\newblock Get Your Vitamin {C}! Robust Fact Verification with Contrastive
  Evidence.
\newblock In \emph{Proceedings of the 2021 Conference of the North American
  Chapter of the Association for Computational Linguistics ({NAACL})}.
\newblock ArXiv:2103.08541.

\bibitem[{Sun et~al.(2026)Sun, Warren, Shklovski, and Augenstein}]{sun2026clue}
Sun, J.; Warren, G.; Shklovski, I.; and Augenstein, I. 2026.
\newblock Explaining Sources of Uncertainty in Automated Fact-Checking.
\newblock In \emph{Proceedings of the 64th Annual Meeting of the Association
  for Computational Linguistics (Volume 1: Long Papers)}, 45510--45534.
  Association for Computational Linguistics.

\bibitem[{Wallat, Nejdl, and Sikdar(2026)}]{wallat2026facts}
Wallat, J.; Nejdl, W.; and Sikdar, S. 2026.
\newblock When Facts Change: Temporal Knowledge Conflict Resolution in {LLM}s.
\newblock In \emph{Findings of the Association for Computational Linguistics:
  {ACL} 2026}, 2154--2184. San Diego, California, United States: Association
  for Computational Linguistics.

\bibitem[{Zheng et~al.(2023)Zheng, Chiang, Sheng, Zhuang, Wu, Zhuang, Lin, Li,
  Li, Xing, Zhang, Gonzalez, and Stoica}]{zheng2023mtbench}
Zheng, L.; Chiang, W.-L.; Sheng, Y.; Zhuang, S.; Wu, Z.; Zhuang, Y.; Lin, Z.;
  Li, Z.; Li, D.; Xing, E.~P.; Zhang, H.; Gonzalez, J.~E.; and Stoica, I. 2023.
\newblock Judging {LLM}-as-a-Judge with {MT}-Bench and Chatbot Arena.
\newblock In \emph{Advances in Neural Information Processing Systems
  ({NeurIPS}) Datasets and Benchmarks Track}.
\newblock ArXiv:2306.05685.

\end{thebibliography}

\appendix
\onecolumn
\section{Prompts}
\label{sec:appendix:prompts}

Each prompt below is the byte content of the template file the runner loads,
named above its listing; the same files are in the code release under
\texttt{scripts/option\_a\_exp/prompts/judges/}, where four prompts are
shipped and the name identifies which one this is.
\verb|<<CLAIM>>| and \verb|<<EVIDENCE>>| are the
substitution points for the claim text and the rendered gold
question--answer evidence block. Line breaks are the template's own, except
where a line exceeds the measure and is wrapped by the typesetter; those
wraps carry a \textcolor{promptrule}{$\hookrightarrow$} marker, so a wrap is
never mistaken for a newline in the prompt.

\subsection{Typed four-option verdict prompt}
\label{sec:appendix:prompts:judge}

This template is shared by all panel judges and by all four contemporary
individual judges. Note the four enumerated trigger conditions
\mbox{(a)--(d)}, the imperative sentence forbidding collapse to
\textsc{Refutes}, and the three demonstrations, each of which resolves
to \textsc{Conflicting}. As discussed in \S\ref{sec:method}, this makes
every directional-overcommitment rate in the paper a conservative lower
bound.

\smallskip
\noindent\texttt{honest\_4opt\_strong.txt}

\begin{lstlisting}[style=promptstyle]
You are a fact-verification judge. Given a claim and supporting evidence, output your verdict.

Use four possible verdicts:
- SUPPORTS: the evidence consistently supports the claim
- REFUTES: the evidence consistently refutes the claim
- NOT_ENOUGH_INFO: the evidence does not address the claim, or is insufficient
- CONFLICTING: prefer this when ANY of the following holds:
  (a) the claim contains multiple sub-parts, AND the evidence supports some sub-parts while refuting others
  (b) the evidence contains both supporting and refuting elements for the same sub-part
  (c) the claim is literally true on its surface but selectively omits context that would change the impression (cherry-picking)
  (d) the claim mixes a true numerical/factual core with an incorrect or misleading framing

When CONFLICTING applies, do NOT collapse it to REFUTES even if the misleading aspect is the dominant impression. Cherry-picking by selective truth is a distinct category from outright falsehood and MUST be reported as CONFLICTING.

Examples (illustrative):
- Claim: "Country X had 75,000 flu deaths in 2019, almost 0 in 2020."
  Evidence: actual 2019 deaths were 21,909; 2020 deaths near zero is correct.
  (*@$\rightarrow$@*) CONFLICTING (one sub-part false, one sub-part true).
- Claim: "Policy Y supports A, B, C, and D."
  Evidence: A is false, B is true, C is unclear, D is partially true.
  (*@$\rightarrow$@*) CONFLICTING (mixed across sub-parts), NOT REFUTES.
- Claim: "Statistic Z is at all-time low this year."
  Evidence: confirms Z is low, but omits seasonal context that would change the interpretation.
  (*@$\rightarrow$@*) CONFLICTING (cherry-picking / misleading framing).

Respond with VALID JSON ONLY in the following exact schema:
{
  "reasoning": "1-3 sentences explaining how the evidence supports your verdict",
  "verdict": "SUPPORTS" | "REFUTES" | "NOT_ENOUGH_INFO" | "CONFLICTING",
  "confidence": <float between 0.0 and 1.0>
}

Do not include any text before or after the JSON object.

Claim: <<CLAIM>>

Evidence:
<<EVIDENCE>>
\end{lstlisting}

\subsection{Certificate validator prompt}
\label{sec:appendix:prompts:validator}

The validator runs independently of the panel and is used only as the
structural channel of the two-channel reference probe
(\S\ref{sec:method:controllers}). Its \texttt{final\_verdict} field is
not consumed by the probe; only the per-subclaim
\texttt{evidence\_state} assignments are.

\smallskip
\noindent\texttt{certificate\_strict\_fewshot.txt}

\begin{lstlisting}[style=promptstyle]
You are a fact-verification judge producing a structured evidence-state certificate.

Given a claim and numbered evidence, decompose the claim into material subclaims and assess the evidence state for each subclaim.

Allowed evidence_state values:
- supports: evidence supports this subclaim
- refutes: evidence refutes this subclaim
- insufficient: evidence does not establish this subclaim either way
- mixed: evidence both supports and refutes this subclaim, or the subclaim is true only with missing context

Final verdict rules:
- SUPPORTS only if every material subclaim is supported and there is no material refuting or mixed evidence.
- REFUTES only if every material subclaim is refuted and there is no material supporting or mixed evidence.
- NOT_ENOUGH_INFO if the material evidence is absent or insufficient and there is no material support/refute mixture.
- CONFLICTING if material support and material refutation coexist, or if any material subclaim is mixed.

Examples:

Example 1:
Claim: "Policy Y supports A, B, C, and D."
Evidence: A is false. B is true. C is unclear. D is partially true.
Output final_verdict: CONFLICTING, because support and refutation coexist across material subclaims.

Example 2:
Claim: "Statistic Z is at an all-time low this year."
Evidence: Z is low this year, but the same seasonal pattern occurs every year and the evidence omits that context.
Output final_verdict: CONFLICTING, because the claim is true only with misleading omitted context.

Example 3:
Claim: "Person A said they would ban technology B."
Evidence: Person A said they would restrict new permits for B, but also said existing uses of B would continue.
Output final_verdict: CONFLICTING, because a blanket ban is refuted but a weaker restriction claim is supported.

Respond with VALID JSON ONLY in the following exact schema:
{
  "subclaims": [
    {
      "text": "atomic material subclaim",
      "material": true,
      "evidence_state": "supports" | "refutes" | "insufficient" | "mixed",
      "evidence_ids": ["E1", "E2"],
      "rationale": "short explanation"
    }
  ],
  "reasoning": "1-3 sentences explaining the final evidence state",
  "final_verdict": "SUPPORTS" | "REFUTES" | "NOT_ENOUGH_INFO" | "CONFLICTING",
  "confidence": <float between 0.0 and 1.0>
}

Do not include any text before or after the JSON object.

Claim: <<CLAIM>>

Numbered evidence:
<<EVIDENCE>>
\end{lstlisting}

\section{Qualitative Case Studies}
\label{sec:appendix:cases}

We provide six representative cases to illustrate the failure modes the
controller catches and the failure modes it does not. The cases are
illustrative, not exhaustive; no statistical rate is derived from the
$N=10$ qualitative case study.

\paragraph{Case 1 --- clean Cherry-pick Override.}
For the claim that US chain migration lets a person bring dozens of
relatives, the cited evidence both refutes the direct reading (one
cannot directly petition an aunt, uncle, or cousin) and supports an
indirect multi-step chain (a naturalized citizen can petition a parent,
who can later petition siblings). The panel issues \textsc{Refutes} at
mean confidence $0.91$; gold is \textsc{Conflicting}. This is a clean CPO
in the sense of Section~\ref{sec:method}: the validator's subclaim
decomposition flags \texttt{material\_mixed}, and the two-channel
controller downgrades to \textsc{Conflicting} regardless of confidence.

\paragraph{Case 2 --- panel amplification on AVeriTeC.}
For the claim ``\emph{Waving the British flag will result in arrest for
breach of the peace}'' (AVeriTeC case 148), the cited evidence has two
strands: a general rule (you cannot be arrested for waving the British
flag) and a Northern-Ireland-specific qualifier (display of flags
including the Union flag has been politically and legally controversial,
with the Police Service of Northern Ireland treating loyalist flag
erection in mixed-population areas as a public-order matter). One typed
judge votes \textsc{Conflicting}, recognizing the qualifier; the other
two judges default to \textsc{Refutes} on the general rule. Majority
voting suppresses the dissent and the panel commits \textsc{Refutes} at
mean confidence $0.89$; gold is \textsc{Conflicting}. The case
exemplifies the population-level finding that panel
$\mathrm{CPO}_{\mathrm{C}}=0.887$ exceeds single-judge $0.840$ on
AVeriTeC (paired-bootstrap CI $[+0.013,+0.080]$, separated from zero);
the effect does not replicate on VitaminC-Mixed under the same
configuration (CI $[-0.060,+0.060]$, straddles zero).

\paragraph{Case 3 --- high-confidence CPO.}
For ``face masks cause hypoxia,'' evidence refutes the claim for healthy
individuals while a separate strand notes prolonged N-95 use in patients
with preexisting lung disease could raise CO\textsubscript{2}. The panel
issues \textsc{Refutes} at mean confidence $0.96$; no confidence
threshold within the explored range ($\tau\le0.95$) blocks the commit.
The validator flags the qualifier subclaim as \texttt{material\_mixed},
and the controller downgrades to \textsc{Conflicting}. The case is the
canonical illustration that confidence and structural evidence provide
complementary information.

\paragraph{Case 4 --- good validator veto.}
For ``Shazad Latif is of Pakistani descent,'' two near-identical evidence
sentences describe the subject as of mixed ``South Asian'' versus mixed
``Pakistani'' descent. The panel issues \textsc{Supports} at confidence
$0.87$, which the confidence channel alone would authorize; the validator
flags \texttt{material\_mixed}, and the controller downgrades to
\textsc{Conflicting}. The auditor's independent reading agrees that a
directional verdict would be misleading. The structural-evidence channel
adds signal the confidence channel cannot recover.

\paragraph{Case 5 --- false validator veto.}
For the claim that Justice Ruth Bader Ginsburg was known for a grueling
fitness regime, the evidence is overwhelmingly supportive (a long-term
trainer, twice-weekly sessions, push-ups). The panel issues
\textsc{Supports} at confidence $0.90$, in agreement with the auditor's
reading, and the confidence channel would have committed correctly. The
validator nonetheless flags a peripheral subclaim as
\texttt{material\_mixed}, and the controller over-blocks to
\textsc{Conflicting}. This is the honest cost of a deterministic veto: it
cannot weight subclaim importance against the main directional question.

\paragraph{Case 6 --- subclaim mismatch.}
For ``for a cumulative 29 of Nigeria's 60 years \dots\ under military
rule,'' the evidence cleanly supports the arithmetic ($13+16=29$). The
validator tags one year-range subclaim as \texttt{material\_mixed},
likely because the evidence does not state the cumulative total
explicitly, and the controller over-blocks despite the directional
verdict being arithmetically correct. Here the validator's
\emph{decomposition} is upstream of its veto rule---a different failure
from Case~5.

Cases~5 and~6 (the validator's failure modes) motivate the follow-up
emphasis on learned evidence-aware controllers and richer
structured-evidence parsers.

\section{Cost-Sensitive Operating-Point Analysis}
\label{sec:appendix:cost}

We convert the same-denominator outcome counts in
Table~\ref{tab:main_results} into per-case expected utility under three
named deployment profiles. The point of the appendix is not to argue
that the controller is preferred in general; it is to make the
deployment trade-off explicit so a reader can see when the
authorization layer is rational and when it is not.

\paragraph{Per-case utility.}
Let an outcome category $o \in \mathcal{O}$ count one of:
$\textsc{correct-S/R}$ (correct directional commit on pure-S/R gold),
$\textsc{correct-Conf}$ (correct \textsc{Conflicting} commit on
gold-\textsc{Conflicting}), $\textsc{wrong-S/R}$ (wrong directional
commit on pure-S/R gold), $\textsc{CPO}$ (directional commit on
gold-\textsc{Conflicting}), $\textsc{false-Conf}$ (\textsc{Conflicting}
commit on pure-S/R gold), $\textsc{No-Commit}$. A profile assigns each
category a cost $w_o$; the per-case utility is the average $w_o$ over
all $N$ cases.

\begin{table}[t]
\centering
\setlength{\tabcolsep}{4pt}
\caption{Three named cost profiles. Profile A treats CPO as $5\times$
costlier than $\textsc{No-Commit}$; Profile C treats it as $12\times$.}
\label{tab:cost_profiles}
\begin{tabular}{lrrrrrr}
\toprule
Profile & correct-S/R & correct-Conf & wrong-S/R & CPO & false-Conf & No-Commit \\
\midrule
A (accuracy)        & $+1$ & $+1$ & $-3$ & $-5$  & $-2$ & $-1$ \\
B (balanced)        & $+1$ & $+1$ & $-5$ & $-8$  & $-2$ & $-1$ \\
C (safety-critical) & $+1$ & $+1$ & $-8$ & $-12$ & $-3$ & $-1$ \\
\bottomrule
\end{tabular}
\end{table}

\paragraph{Per-system results.}
Table~\ref{tab:cost_results} reports per-case utility for each
controller under each profile on AVeriTeC ($N=285$) and VitaminC-Mixed
($N=250$). F~$\tau{=}0.90$ is the per-case utility maximiser on every
(dataset $\times$ profile) cell at these profile settings. The
mechanism is structural rather than driven by the CPO penalty
magnitude: Stage~1 converts CPO events into
$\textsc{correct-Conf}$ credit ($+1$) rather than $0$; at cost
ratios where $\textsc{wrong-S/R}$ and $\textsc{CPO}$ are equal,
typed direct (L1) becomes competitive on AVeriTeC because it commits
more often.

\begin{table}[t]
\centering
\small
\setlength{\tabcolsep}{4pt}
\caption{Per-case utility ($\uparrow$ better) on AVeriTeC ($N=285$) and
VitaminC-Mixed ($N=250$) for the three profiles defined in
Table~\ref{tab:cost_profiles}.}
\label{tab:cost_results}
\begin{tabular}{llrrr}
\toprule
Dataset & System & A & B & C \\
\midrule
AVeriTeC & typed-direct (L1)       & $-0.158$ & $-0.502$ & $-1.046$ \\
         & validator-only (L4)     & $-0.011$ & $-0.197$ & $-0.583$ \\
         & E $\tau$=0.85 (L3)      & $-0.109$ & $-0.404$ & $-0.881$ \\
         & E $\tau$=0.90 (L3)      & $-0.039$ & $-0.214$ & $-0.526$ \\
         & F $\tau$=0.85 (L5)      & $+0.018$ & $-0.140$ & $-0.488$ \\
         & \textbf{F $\tau$=0.90 (L5)} & $\mathbf{+0.032}$ & $\mathbf{-0.070}$ & $\mathbf{-0.340}$ \\
\midrule
VitaminC & typed-direct (L1)       & $-0.596$ & $-1.196$ & $-2.064$ \\
         & validator-only (L4)     & $-0.460$ & $-0.944$ & $-1.684$ \\
         & E $\tau$=0.85 (L3)      & $-0.572$ & $-1.132$ & $-1.940$ \\
         & E $\tau$=0.90 (L3)      & $-0.524$ & $-1.004$ & $-1.700$ \\
         & F $\tau$=0.85 (L5)      & $-0.436$ & $-0.888$ & $-1.580$ \\
         & \textbf{F $\tau$=0.90 (L5)} & $\mathbf{-0.428}$ & $\mathbf{-0.832}$ & $\mathbf{-1.456}$ \\
\bottomrule
\end{tabular}
\end{table}

\paragraph{Sensitivity to the asymmetry.}
The asymmetry between $\textsc{wrong-S/R}$ and $\textsc{CPO}$ costs is
the load-bearing assumption of all three profiles in
Table~\ref{tab:cost_profiles}: a deployment that values the two equally
($w_{\textsc{wrong-S/R}}=w_{\textsc{CPO}}$) collapses the L5 advantage on
AVeriTeC, because L1's higher commit rate compensates for its higher CPO.
The cost-sensitive analysis is therefore deployment-conditional: it
identifies a regime in which an authorization layer is rational, not a
universal recommendation.

\section{Cross-Dataset Replication: VitaminC-Mixed}
\label{sec:appendix:vitaminc}

We mirror Table~\ref{tab:main_results}'s ladder layout on the VitaminC-Mixed
substrate ($N=250$; $N_{\mathrm{S/R}}=100$, $N_{\mathrm{C}}=100$). The
VitaminC-Mixed pool is built by concatenating the supporting and refuting
sentence in a VitaminC contrastive pair to synthesise an evidence set
that is mixed by construction. Mixed-evidence claims here therefore have
a different distributional character than AVeriTeC's natural
Conflicting/Cherrypicking subset, and this appendix isolates which
findings transfer.

\begin{table}[t]
\centering
\setlength{\tabcolsep}{4pt}
\caption{Cross-dataset replication on VitaminC-Mixed
($N{=}250$). Metric definitions match Table~\ref{tab:main_results}.
L0 single-3-opt and panel-3-opt are reported only on the conflicting
subset for $\mathrm{CPO}_{\mathrm{C}}$ and $\mathrm{Rec}_{\mathrm{C}}$;
the L0 panel and single rows tie ($\mathrm{CPO}_{\mathrm{C}}=0.77$),
showing the AVeriTeC panel-amplification effect does not transfer here.}
\label{tab:vitaminc_results}
\begin{tabular}{lrrrrrr}
\toprule
System & $\mathrm{Cov}$ & $\mathrm{SE}$ & $\mathrm{CPO}_N$ & $\mathrm{CPO}_{\mathrm{C}}$ & $\mathrm{Acc}_{\mathrm{S/R}}$ & $\mathrm{Rec}_{\mathrm{C}}$ \\
\midrule
\multicolumn{7}{l}{\emph{L0 --- three-option direct}} \\
Single 3-opt        & --- & --- & --- & 0.770 & --- & 0.000 \\
Panel 3-opt         & --- & --- & --- & 0.770 & --- & 0.000 \\
\midrule
\multicolumn{7}{l}{\emph{L1 --- typed panel direct (A)}} \\
Panel + typed       & 0.536 & 0.448 & 0.120 & 0.300 & 0.740 & 0.700 \\
\midrule
\multicolumn{7}{l}{\emph{L3 --- confidence-only (E)}} \\
$\tau$=0.85        & 0.508 & 0.433 & 0.120 & 0.300 & 0.720 & 0.700 \\
$\tau$=0.90        & 0.452 & 0.416 & 0.104 & 0.260 & 0.660 & 0.700 \\
$\tau$=0.95        & 0.208 & 0.212 & 0.024 & 0.060 & 0.410 & 0.700 \\
\midrule
\multicolumn{7}{l}{\emph{L4 --- validator-veto only (D)}} \\
conflict\_only     & 0.476 & 0.412 & 0.092 & 0.230 & 0.700 & 0.770 \\
\midrule
\multicolumn{7}{l}{\emph{L5 --- two-channel reference probe (F)}} \\
$\tau$=0.85        & 0.456 & 0.395 & 0.092 & 0.230 & 0.690 & 0.770 \\
$\tau$=0.90        & 0.412 & 0.388 & 0.084 & 0.210 & 0.630 & 0.770 \\
$\tau$=0.95        & 0.204 & 0.216 & 0.024 & 0.060 & 0.400 & 0.770 \\
\bottomrule
\end{tabular}
\end{table}

\paragraph{What transfers.}
(i)~The vocabulary-fix step (L0$\to$L1) reduces conditional CPO from
$0.770$ to $0.300$ --- same direction as on AVeriTeC, though residual
remains. (ii)~Across L1$\to$L3$\to$L5 the operating-point pattern at
$\tau{=}0.90$ matches AVeriTeC: L5 reduces $\mathrm{SE}$ and
$\mathrm{CPO}_N$ over L3 at comparable coverage. (iii)~The random-veto
control (Figure~\ref{fig:random_veto} on AVeriTeC) re-run on VitaminC
again places L5 at empirical $p<1/2001$ on
$\mathrm{Acc}_{\mathrm{S/R}}$ and $\mathrm{Rec}_{\mathrm{C}}$ at matched
coverage (see Section~\ref{sec:ablation}).

\paragraph{What does not transfer.}
(i)~Panel amplification at L0 (AVeriTeC: $+0.047$, CI $[+0.013,+0.080]$)
does not replicate on VitaminC ($0.000$, CI $[-0.060,+0.060]$). The
synthetic construction of VitaminC-Mixed concatenates a single
supporting and a single refuting sentence, and the typed panel agrees
unanimously on directional commitment at this level of evidence
homogeneity. (ii)~At the extreme low-coverage regime ($\tau{=}0.95$),
L3 and L5 tie on every metric; both reach $\mathrm{CPO}_N{=}0.024$ at
$\mathrm{Cov}{\approx}0.20$, where confidence alone already filters the
remaining mixed-evidence commits.

\paragraph{Panel amplification anatomy.}
Table~\ref{tab:amplification_anatomy} reports the per-judge vote
distribution on the gold-\textsc{Conflicting} subsets of both datasets,
which Section~\ref{sec:results:l2} cites inline. Under the 3-opt
schema, roughly two-thirds of conflicting cases have all three judges
voting directionally on both datasets, so amplification is dominated
by shared directional bias rather than majority overriding minority
dissent. Under the 4-opt typed schema, VitaminC actually has more
suppressed-dissent cases (22\% vs 9\%) but also far fewer
unanimous-directional cases, so panel CPO does not amplify.

\begin{table}[h]
\centering
\footnotesize
\setlength{\tabcolsep}{4pt}
\caption{Panel amplification anatomy on the gold-\textsc{Conflicting}
subset of each dataset. Rows partition the subset by how many of the
three judges voted directionally, so each block sums to $100\%$ up to
rounding. ``Unanimous directional'' is a $3/0$ split; ``dissent
suppressed'' is $2$ directional $+\,1$ non-directional, where the panel
majority commits despite a minority vote to withhold; ``minority
directional'' is $1{+}2$, where the same aggregation rule correctly
discards the directional vote. ``Other'' is cases with an incomplete
vote record.}
\label{tab:amplification_anatomy}
\begin{tabular}{lrr}
\toprule
Vote pattern (gold $=$ \textsc{Conflicting}) & AVeriTeC & VitaminC \\
\midrule
\multicolumn{3}{l}{\emph{3-opt panel ($N_{\mathrm{C}}{=}150 / 100$)}} \\
unanimous directional ($3/0$)                 & 70.7\% & 66.0\% \\
dissent suppressed ($2$ dir $+$ 1 NEI)        & 18.0\% & 11.0\% \\
minority directional ($1$ dir $+$ 2 NEI)      &  4.7\% & 16.0\% \\
no directional vote ($0/3$)                   &  4.7\% &  7.0\% \\
other (incomplete votes)                      &  2.0\% & --- \\
\multicolumn{3}{l}{\emph{4-opt typed panel}} \\
unanimous directional ($3/0$)                 &  9.3\% & 13.0\% \\
\textsc{Conflicting} dissent suppressed       &  9.3\% & 22.0\% \\
minority directional ($1$ dir $+$ 2 non-dir)  & 24.0\% & 24.0\% \\
no directional vote ($0/3$)                   & 56.0\% & 41.0\% \\
other (incomplete votes)                      &  1.3\% & --- \\
\bottomrule
\end{tabular}\\[3pt]
\begin{minipage}{\columnwidth}
\footnotesize Percentages here are over the full gold-\textsc{Conflicting}
subset, so AVeriTeC's $9.3\%$ suppressed-dissent row is $14$ of $150$ cases.
\S\ref{sec:results:l2} instead reports $13$ of $27$, because it conditions on
the panel's CPO commits and drops the one whose vote record is incomplete:
the $3/0$ versus $2/1$ anatomy is undefined for that case.
\end{minipage}
\end{table}

\end{document}